\documentclass[11pt]{article}

\usepackage{fullpage}
\usepackage{framed}
\usepackage{mdframed}
\usepackage{enumerate}
\usepackage[inline]{enumitem}
\usepackage[T1]{fontenc}
\usepackage{moresize}
\usepackage{bm}
\usepackage{dsfont}
\usepackage{amsmath}
\usepackage{amssymb}
\usepackage{amsthm}
\usepackage{amsfonts}
\usepackage{stmaryrd}
\usepackage{array}
\usepackage{mathrsfs}
\usepackage{mathtools} 
\usepackage{extarrows}
\usepackage{stackrel}
\usepackage[nodisplayskipstretch]{setspace}
\usepackage{color}
\usepackage[usenames,dvipsnames]{xcolor}
\usepackage{cancel}
\usepackage{soul}
\usepackage{xfrac}
\usepackage{siunitx}
\usepackage{graphicx}
\usepackage{float}
\usepackage{rotating}
\usepackage{subcaption}
\usepackage{overpic}
\usepackage[all]{xy}
\usepackage{bm}
\usepackage{tikz}
\usetikzlibrary{arrows,matrix,positioning,calc,automata,patterns}
\usepackage{booktabs}
\usepackage{dcolumn}
\usepackage{multirow}
\usepackage{diagbox}
\usepackage{verbatim}
\usepackage{listings}
\usepackage{algorithm}
\usepackage{algpseudocode}
\usepackage{fancyvrb}
\usepackage{hyperref}
\usepackage[round]{natbib}
\hypersetup{
    unicode=false,          
    pdftoolbar=true,        
    pdfmenubar=true,        
    pdffitwindow=false,     
    pdfstartview={FitH},    
    pdftitle={My title},    
    pdfauthor={Author},     
    pdfsubject={Subject},   
    pdfcreator={Creator},   
    pdfproducer={Producer}, 
    pdfkeywords={key1, key2}, 
    pdfnewwindow=true,      
    colorlinks=true,        
    linkcolor=blue,         
    citecolor=blue,         
    filecolor=blue,         
    urlcolor=cyan           
}

\def\tr{\mathop{\text{tr}}\kern.2ex}

\newcommand{\argmax}{\mathop{\mathrm{argmax}}}

\newcolumntype{L}[1]{>{\raggedright\let\newline\\\arraybackslash\hspace{0pt}}m{#1}}
\newcolumntype{C}[1]{>{  \centering\let\newline\\\arraybackslash\hspace{0pt}}m{#1}}
\newcolumntype{R}[1]{>{ \raggedleft\let\newline\\\arraybackslash\hspace{0pt}}m{#1}}
\newcolumntype{d}[1]{D{.}{.}{#1}}
\newcolumntype{H}{>{\setbox0=\hbox\bgroup}c<{\egroup}@{}}
\newcolumntype{Z}{>{\setbox0=\hbox\bgroup}c<{\egroup}@{\hspace*{-\tabcolsep}}}
\numberwithin{equation}{section}

\newtheorem{lemma}{Lemma}[section]
\newtheorem{proposition}{Proposition}[section]
\newtheorem{assumption}{Assumption}[section]

\newtheorem{definition}{Definition}[section]
\newtheorem{innercustomgeneric}{\customgenericname}

\providecommand{\customgenericname}{}
\newcommand{\newcustomtheorem}[2]{%
  \newenvironment{#1}[1]
  {%
   \renewcommand\customgenericname{#2}%
   \renewcommand\theinnercustomgeneric{##1}%
   \innercustomgeneric
  }
  {\endinnercustomgeneric}
}
\newcustomtheorem{customtheorem}{theorem}
\newcustomtheorem{customassumption}{Assumption}
\newcustomtheorem{customlemma}{lemma}
\newcustomtheorem{customexample}{example}
\theoremstyle{definition}
\newtheorem{example}{Example}[section]

\allowdisplaybreaks

\begin{document}

\setlength{\abovedisplayskip}{5pt}
\setlength{\belowdisplayskip}{5pt}
\setlength{\abovedisplayshortskip}{5pt}
\setlength{\belowdisplayshortskip}{5pt}

\title{Bayesian Finite Mixtures of Ising Models}
\author{
Zhen Miao\thanks{Department of Statistics, University of Washington, Seattle; e-mail: {\tt zhenm@uw.edu}},~
Yen-Chi Chen\thanks{Department of Statistics, University of Washington, Seattle; e-mail: {\tt yenchic@uw.edu}},~
Adrian Dobra\thanks{Department of Statistics, University of Washington, Seattle; e-mail: {\tt adobra@uw.edu}}
}
\maketitle

\begin{abstract}
We introduce finite mixtures of Ising models as a novel approach to study multivariate patterns of associations of binary variables. Our proposed models combine the strengths of Ising models and multivariate Bernoulli mixture models. We examine conditions required for the identifiability of Ising mixture models, and develop a Bayesian framework for fitting them. Through simulation experiments and real data examples, we show that Ising mixture models lead to meaningful results for sparse binary contingency tables.
\end{abstract}

\section{Introduction}
Loglinear models have been widely used in the analysis of multivariate categorical data in many scientific fields, such as biological sciences, natural language processing, data mining \citep{bishop1975discrete,christensen1997log-linear} due to their ability to capture first, second and higher order interactions among the observed variables. They originated from testing for the absence of interactions in $2\times2\times2$ contingency tables \citep{bartlett1935contingency}, and were later generalized to multidimensional contingency tables \citep{roy1956hypothesis,darroch1962interactions,good1963maximum,goodman1963methods}. In this paper we focus on Ising models \citep{kindermannsnell1980} which can be viewed as graphical loglinear models for binary variables that include only first order interaction terms \citep{lauritzen1996}. Ising models have particular relevance in network analysis of binary data \citep{vanborkulo2014}.

Various frequentist approaches for estimation of loglinear models have been proposed in the literature. For example, the existence and uniqueness of maximum likelihood estimators (MLEs) have been studied for different types of tables, e.g., three-way contingency tables \citep{birch1963maximum},  
general contingency tables \citep{haberman1974analysis,aickin1979existence,verbeek1992compactification}, and 
sparse contingency tables \citep{fienberg2012maximum}. The computation of the MLEs can be performed using closed-form expressions \citep[Chapter 3.4]{bishop1975discrete}, via iterative proportional fitting based on matrix inversion techniques \citep{goodman1964simple} or Newton-Raphson techniques \citep{haberman1974analysis}.
\citet{fienberg2000contingency} provides a comprehensive review. Bayesian approaches for loglinear modelling have also received a lot of interest, with a particular focus on the development of suitable prior distributions. Key examples include the multivariate normal prior \citep{knuiman1988incorporating,dellaportas1999markov,brooks2001prior,dobra-et-2006}, the spike-and-slab prior \citep{rockova2018particle},
the hyper Dirichlet conjugate prior \citep{dawid1993hyper-markov} and its generalization, as well as the Diaconis–Ylvisaker (DY) conjugate prior \citep{massam2009conjugate}. The use of the DY prior for model selection has been studied  in \citet{dobra2010mode,letac2012bayes}.

Sparse contingency tables raise key issues related to estimation and fit for Ising models as well as for the larger family of loglinear models. Scalability of Ising models to discrete datasets with many variables has been solved in various ways \--- see, among others, \citet{ravikumar2010highdimensional}. Sparse contingency tables are also characterized by the imbalance of their cell counts \citep{dobralenkoski2011}. Ising models together with the richer class of loglinear models tend to oversmooth the fitted cell probabilities which makes them unable to capture the magnitude of the larger cell counts. To this end, several classes of mixture models have been proposed as alternatives. Multivariate Bernoulli mixture models \citep{carreira2000practical,allman2009identifiability} have shown promising results in applications \citep{juan2002use,juan2004bernoulli}. Other classes of mixture models for categorical data include parallel factor analysis (PARAFAC) \citep{bro1997parafac}, simplex factor models \citep{bhattacharya2012simplex}, sparse PARAFAC \citep{zhou2015bayesian}, the Tucker decomposition \citep{deLathauwer2000multilinear}, and the collapsed Tucker decomposition \citep{johndrow2017tensor}. While these mixture models perform well with respect to fit for sparse contingency tables, they are not easily interpretable especially when it comes to inferring relevant patterns of multivariate associations among discrete variables. We note two studies exploring the connections between the two modeling paradigms. \citet{papathomas2016exploring}  leverages mixture models to reduce the number of parameters in a loglinear model, while \citet{johndrow2017tensor} delves into the relationship between dimension reduction in mixture models and loglinear models.

A common issue in finite mixture models is their identifiability \citep{teicher1967identifiability,yakowitz1968identifiability,titterington1985statistical}.
Identifiability has been studied for various types of mixture models such as uniform mixtures and binomial mixtures \citep{teicher1961identifiability}, normal mixtures, exponential mixtures, Gamma mixtures \citep{teicher1963identifiability}, Poisson mixtures \citep{teicher1960mixture}, and negative binomial mixtures \citep{titterington1985statistical}.
For general mixture models, \cite{teicher1963identifiability} suggests using moment generating functions to prove identifiability, and \cite{teicher1967identifiability} presents sufficient conditions for the mixture of product densities.
The most recent identifiability results are for multivariate Bernoulli mixture models through the conditional independence assumption \citep{allman2009identifiability,xu2017identifiability}.

The novel contributions of our work are as follows. In Section~\ref{subsec:bayesianMethods}, we propose a novel Bayesian method for fitting finite mixtures of Ising models with the modeling goal of inferring associations between binary variables. In Sections~\ref{subsec:simulationLogLinear} and~\ref{subsec:applicationlogLinear}, our novel framework is illustrated through simulation experiments and real data applications.
We provide sufficient and necessary conditions for identifiability of the model in Section~\ref{subsec:identifiability} with proofs in Section~\ref{sec:proofs}. Finally, in Section~\ref{sec:discussion}, we discuss our results together with several potential extensions.

\section{Ising Mixture Models}
\label{subsec:bayesianMethods}

\subsection{Notation}
Let $\bm X:=(X_1,\ldots,X_d)^T$ be a  vector of $d\in\mathbb{N}^+$ binary random variables, each taking values of $0$ or $1$.
The set of possible values for $\bm X$ is denoted as $I:=\{0,1\}^d$ with elements $\bm i = (i_1,\ldots,i_d)\in I$ assumed to be order lexicographically. The vector of cell probabilities of $\bm X$ are $(P(\bm X=\bm i):\bm i\in I)^T$.

Let the main effect of $X_v$ be denoted by $\theta_v\in\mathbb{R}$, where $v\in[d]=\{1,2,\ldots,d\}$. The interaction effect between $X_{v^\prime}$ and $X_v$ is denoted by $\theta_{v^\prime v}\in\mathbb{R}$, where $v^\prime < v$ and both $v^\prime, v \in [d]$. We say that ${\bm X}$ follows an Ising model if 
the logarithm of the probability associated with cell $\bm i \in I$ is proportional to a linear combination of main effects $(\theta_v:v\in[d])^T$ and interaction effects $(\theta_{v^\prime v}:v^\prime<v)^T$:
\[
\log p_{\bm i}(\bm\theta)
=\sum_{v=1}^d\theta_vi_v+\sum_{v^\prime=1}^{d-1}\sum_{v=v^\prime+1}^d\theta_{v^\prime v}i_{v^\prime}i_v+C(\bm\theta),
\]
where $C(\bm\theta)$ is the logarithm of the normalizing constant, $\bm\theta$ represents the union of main effects and interaction effects, i.e., $\bm\theta=(\theta_1,\ldots,\theta_d,\theta_{12},\ldots,\theta_{(d-1)d})$ and $p_{\bm i}(\bm\theta) = P({\bm X} = {\bm i}\mid \bm\theta$). If $\theta_{v^\prime v}=0$, variables $X_{v^\prime}$ and $X_v$ are conditionally independent given the rest. The main effects and the interaction terms can be interpreted as conditional log odds and log odds ratios \citep{agresti2002}.

The Ising model can be expressed as
\begin{equation}
\bm p(\bm\theta)
=(p_{\bm i}(\bm\theta):\bm i\in I)^T
=\exp(A^T\bm\theta)/[\bm 1_{|I|}^T\exp(A^T\bm\theta)], \label{eq:ising}    
\end{equation}
where $\bm 1_{|I|}$ is a $|I|$-dimensional constant vector of all ones, 
$A\in\mathbb{R}^{[d(d+1)/2]\times|I|}$ is a conventionally defined constant design matrix \citep{wang2019approximating}. Here applying the exponential function $\exp(\cdot)$ to a vector means applying it element-wise to obtain a vector. We illustrate the definition of $A$ through an example.

\begin{example}
Suppose we have two binary random variables $\bm X=(X_1,X_2)^T$ that follow an Ising model with main effects $\theta_1,\theta_2$ and interaction effect $\theta_{12}$.
The following linear combination
\[
\log p_{\bm i}=\log p_{(i_1,i_2)}=\theta_1i_1+\theta_2i_2+\theta_{12}i_1i_2+C(\bm\theta),
\]
for $\bm i\in I=\{(0,0),(1,0),(0,1),(1,1)\}$,
is equivalent to
\begin{align*}
\begin{pmatrix}
\log p_{(0,0)}\\
\log p_{(1,0)}\\
\log p_{(0,1)}\\
\log p_{(1,1)}
\end{pmatrix} - C(\bm\theta)
=
\begin{pmatrix}
\theta_1\cdot0+\theta_2\cdot0+\theta_{12}\cdot0\cdot0\\
\theta_1\cdot1+\theta_2\cdot0+\theta_{12}\cdot1\cdot0\\
\theta_1\cdot0+\theta_2\cdot1+\theta_{12}\cdot0\cdot1\\
\theta_1\cdot1+\theta_2\cdot1+\theta_{12}\cdot1\cdot1
\end{pmatrix}
=
\begin{pmatrix}
0\\
\theta_1\\
\theta_2\\
\theta_1+\theta_2+\theta_{12}
\end{pmatrix}
=A^T
\begin{pmatrix}
\theta_1\\
\theta_2\\
\theta_{12}
\end{pmatrix}
,
\end{align*}
where $A^T=[[0,0,0],[1,0,0],[0,1,0],[1,1,1]]$.
\end{example}

We say $X$ follows an Ising mixture model if its vector of cell probabilities is expressed as a finite mixture of Ising models, i.e.
\begin{equation}
(P(\bm X=\bm i):\bm i\in I)^T=\bm p_{\rm mix}(\bm w,\bm\Theta)
:=(p_{{\rm mix},\bm i}(\bm w,\bm\Theta),\bm i\in I)^T
:=
\sum_{k=1}^Kw^{(k)}\bm p(\bm\theta^{(k)}), \label{eq:isingmixture}    
\end{equation}
with
\[
\bm p(\bm\theta^{(k)})=\exp(A^T\bm\theta^{(k)})/[\bm 1_{|I|}^T\exp(A^T\bm\theta^{(k)})].
\]
Here $K\in\mathbb{N}^+$ is the number of components, $\bm w=(w_k:k\in[K])^T\in(0,1)^K$ represents the weights of the $K$ components with $\sum_{k\in[K]}w_k=1$, $\bm\theta^{(k)}:=(\theta_1^{(k)},\ldots,\theta_d^{(k)},\theta_{12}^{(k)},\ldots,\theta_{(d-1)d}^{(k)})\in\mathbb{R}^{d(d+1)/2}$ is the vector of main effects and interaction effects for component $k$. We define the vector of parameters of the Ising mixture model as $\bm\Theta:=(\bm\theta^{(k)},k\in[K])\in\mathbb{R}^{Kd(d+1)/2}$.

The Ising mixture model says that the binary random vector $\bm{X}$ is drawn from $K$ subpopulations with probabilities $\bm{w}$. Given the $k$-th subpopulation, $k\in[K]$, $\bm{X}$ follows an Ising model with parameters $\bm{\theta}^{(k)}$.

We assume that the observed data consist of $N$ i.i.d. observations of $\bm X$ under the simple multinomial sampling theme \citep{cochran1952chisquare}. It then follows that the resulting cell counts ${\bm n}:=(n_{\bm i}:{\bm i} \in I)^T$ follow a Multinomial($N,\bm p$) distribution, where $N=\sum_{\bm i\in I}n_{\bm i}$. Let $\|\cdot\|_2$ be the Euclidean norm of a vector.

\subsection{Prior specification}
In our proposed Bayesian framework, we assume that the mixture weights $\bm w$ follow a Dirichlet distribution ${\rm Dirichlet}(\bm \alpha)$ with parameters $\bm\alpha:=(\alpha^{(k)},k\in[K])$ with $\alpha^{(k)}>0$. This is a common choice for nonnegative parameters that sum to 1 \citep{olkin1964multivariate}. For each component $k$, we assume that the main effects $(\theta_{v}^{(k)},v\in[d])^T$ independently and identically follow a normal distribution with mean $0$ and variance $\sigma_1^2>0$, denoted by $N(0,\sigma_1^2)$. We also assume that the interaction effects $(\theta_{v^\prime v}^{(k)},v^\prime<v)^T$ independently and identically follow a continuous spike-and-slab prior with spike variance $\sigma_0^2$ and slab variance $\sigma_1^2$ where $0<\sigma_0<\sigma_1$. Specifically, 
\begin{equation}
\theta_{v^\prime v}^{(k)}|\gamma_{v^\prime v}^{(k)}\sim(1-\gamma_{v^\prime v}^{(k)})N(0,\sigma_0^2)+\gamma_{v^\prime v}^{(k)}N(0,\sigma_1^2), \label{eq:spikeandslab}
\end{equation}
where $\gamma_{v^\prime v}^{(k)}$ is the indicator of the association between variables $X_{v^\prime}$ and $X_v$ in the $k$-th component, and it is assumed to follow a Bernoulli distribution with known parameter $\beta\in(0,1)$. More precisely, $\gamma_{v^\prime v}^{(k)} = 0$ indicates that the interaction effect between variables $v^\prime$ and $v$ in the $k$th component is more likely to be close to $0$, while a value of $\gamma_{v^\prime v}^{(k)} = 1$ implies that this interaction effect is more likely to be non-zero. This is particularly clear when the spike variance, $\sigma_0^2$, is set to zero, resulting in a point-mass mixture of a point mass at $0$ and a normal distribution for the interaction effects. Denote $\bm\gamma^{(k)}:=(\gamma_{v^\prime v}^{(k)},v^\prime<v)^T$ and let 
$\bm\Gamma:=(\bm\gamma^{(k)},k\in[K])\in\{0,1\}^{Kd(d-1)/2}$. The collection of binary random variables $\bm\Gamma$ is of key interest since it represents the presence of non-zero interaction effects between variables in each component.
Moreover, We use a directed acyclic graph to illustrate the associations between parameters in our model specification, see Figure~\ref{fig:DAGforParameters}.

The Normal distributions in the spike-and-slab prior can be changed to other distributions such as the Laplace distribution. The continuous spike-and-slab prior, which serves as the predecessor to the point-mass mixture, has been gaining renewed attention in recent years \citep{rockova2014emvs,rockova2018spike}. The continuity of this prior allows for a more fluid exploration of posteriors using both MCMC and optimization techniques due to its ability to decrease the spike variance to zero and explore the entire path of posteriors as it approaches the point mass mixture. In addition to computational benefits, continuous mixture priors can also exhibit optimal posterior behavior, such as the oracle property of the posterior mean \citep{ishwaran2005spike,rockova2018particle}.

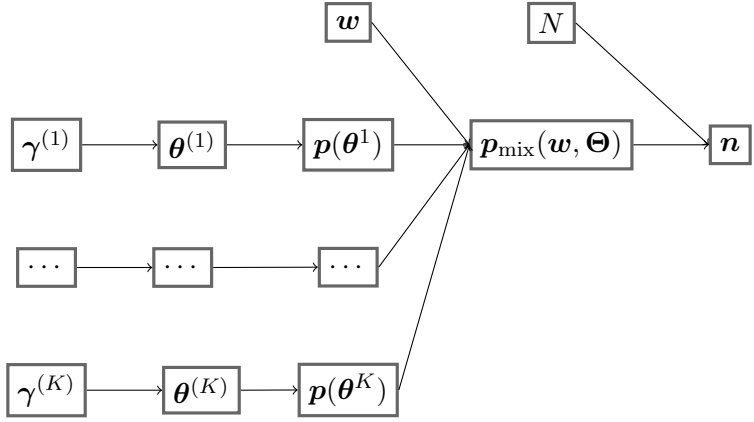
\begin{figure}
\centering
\begin{tikzpicture}[
squarednode/.style={rectangle, draw=black!60, very thick, minimum size=5mm},
]
\node[squarednode](gamma1){$\bm\gamma^{(1)}$};
\node[squarednode](gammak)[below=of gamma1]{$\cdots$};
\node[squarednode](gammaK)[below=of gammak]{$\bm\gamma^{(K)}$};

\node[squarednode](theta1)[right=of gamma1]{$\bm\theta^{(1)}$};
\node[squarednode](thetak)[right=of gammak]{$\cdots$};
\node[squarednode](thetaK)[right=of gammaK]{$\bm\theta^{(K)}$};

\node[squarednode](p1)[right=of theta1]{$\bm p(\bm\theta^{1})$};
\node[squarednode](pk)[below=of p1]{$\cdots$};
\node[squarednode](pK)[below=of pk]{$\bm p(\bm\theta^{K})$};
\node[squarednode](w)[above=of p1]{$\bm w$};

\node[squarednode](pmix)[right=of p1]{$\bm p_{\rm mix}(\bm w,\bm\Theta)$};
\node[squarednode](N)[above=of pmix]{$N$};
\node[squarednode](n)[right=of pmix]{$\bm n$};
\draw[->] (gamma1.east) -- (theta1.west);
\draw[->] (gammak.east) -- (thetak.west);
\draw[->] (gammaK.east) -- (thetaK.west);
\draw[->] (theta1.east) -- (p1.west);
\draw[->] (thetak.east) -- (pk.west);
\draw[->] (thetaK.east) -- (pK.west);
\draw[->] (p1.east) -- (pmix.west);
\draw[->] (pk.east) -- (pmix.west);
\draw[->] (pK.east) -- (pmix.west);
\draw[->] (w.east) -- (pmix.west);
\draw[->] (N.east) -- (n.west);
\draw[->] (pmix.east) -- (n.west);
\end{tikzpicture}
\caption{A directed acyclic graph to illustrate the associations between parameters in our model specification.}
\label{fig:DAGforParameters}
\end{figure}

\subsection{Posterior distribution}

To obtain the joint posterior distribution of $\bm w$, $\bm\Theta$ and $\bm\Gamma$, we begin by discussing the Ising model \eqref{eq:ising}, and then consider the Ising mixture model \eqref{eq:isingmixture} with $K\ge 2$ components.

\subsubsection{The Ising model}
Under Multinomial sampling, we have ${\bm n}\mid \bm\theta\sim{\rm Multinomial}(N,\bm p(\bm\theta))$.
It then follows that the probability mass function of $\bm n$ given $\bm\theta$ is
\begin{align*}
\pi({\bm n}\mid \bm\theta)
=\frac{N!}{\prod_{\bm i\in I} n_{\bm i}!}[{\bm 1}_{|I|}\cdot\exp(A^{T}\bm\theta)]^{-N}\exp(\bm n^TA^T\bm\theta)
=\frac{N!}{\prod_{\bm i\in I} n_{\bm i}!}
\exp[N\ell(\bm\theta|\bm n)],
\end{align*}
where $\ell(\bm\theta\mid \bm n):=-\log[{\bm 1}_{|I|}\cdot\exp(A^{T}\bm\theta)]+\bm n^TA^T\bm\theta/N$ is the log-likelihood function of the Ising model. Thus, the joint distribution 
$\pi(\bm n,\bm\theta,\bm\gamma)
=\pi({\bm n}|\bm\theta)\pi(\bm\theta|\bm\gamma)\pi(\bm\gamma)$ 
can be further written as
\begin{align*}
\pi(\bm n,\bm\theta,\bm\gamma)
&=\frac{N!}{\prod_{\bm i\in I} n_{\bm i}!}\exp[N\ell(\bm\theta)]
\cdot\prod_{v\in[d]}
\frac{1}{\sigma_1\sqrt{2\pi}}\exp\Big(-\frac{\theta_{v}^2}{2\sigma_1^2}\Big)\\
&\quad\cdot
\prod_{v^\prime<v}\Big(\frac{1}{\sigma_0\sqrt{2\pi}}\exp\Big(-\frac{\theta_{v^\prime v}^2}{2\sigma_0^2}\Big)\Big)^{1-\gamma_{v^\prime v}}
\Big(\frac{1}{\sigma_1\sqrt{2\pi}}\exp\Big(-\frac{\theta_{v^\prime v}^2}{2\sigma_1^2}\Big)\Big)^{\gamma_{v^\prime v}}\\
&\quad\cdot
\prod_{v^\prime<v}\beta^{\gamma_{v^\prime v}}(1-\beta)^{1-\gamma_{v^\prime v}}.
\end{align*}
After some algebra, we obtain
\begin{align}
\pi(\bm\gamma\mid \bm n)
&\propto
\int
\exp\left[N\ell(\bm\theta\mid\bm n)-\frac{\sum_{v\in[d]}\theta_{v}^2}{2\sigma_1^2}\right]\cdot
\prod_{v^\prime<v}\left[\exp\left(\frac{-\theta_{v^\prime v}^2}{2\sigma_0^2}\right)\right]^{1-\gamma_{v^\prime v}}\left[\frac{\beta\sigma_0}{(1-\beta)\sigma_1}\exp\left(\frac{-\theta_{v^\prime v}^2}{2\sigma_1^2}\right)\right]^{\gamma_{v^\prime v}}{\sf d}\bm\theta.
\label{eq:posteriorOfGamma}
\end{align}
We note that the posterior distribution of $\bm\gamma$ is a mixture of Bernoulli distributions.
Specifically, given $\bm\theta$, the posterior distribution of $\bm\gamma$ is 
\[
\gamma_{v^\prime v}|\bm\theta\sim{\rm Bernoulli}\left(\frac{1}{1+(1-\beta)\sigma_1/(\beta\sigma_0)\cdot e^{\theta_{v^\prime v}^2(1/\sigma_1^2-1/\sigma_0^2)/2}}\right),v^\prime<v,~{\rm independently}.
\]
As a consequence, the posterior mean of $\bm\gamma$ can be written as 
$
E_{\theta_{v^\prime v}\sim \pi(\bm\theta\mid\bm n)}[r(\theta_{v^\prime v})],
v^\prime<v,
$
where 
\begin{align}
\theta\mapsto 
r(\theta):=\frac{1}{1+(1-\beta)\sigma_1/(\beta\sigma_0)\cdot e^{\theta^2(1/\sigma_1^2-1/\sigma_0^2)/2}},
\label{eq:r(theta)}
\end{align} 
For any fixed $\sigma_1>0$, we have $r(\theta)\to I(\theta\neq0)$ as $\sigma_0\to0$ for any $\theta\in\mathbb{R}$.
This means that $r(\theta)$ can be interpreted as a smoothed measurement of whether the interaction effect $\theta$ is zero. For a small $\sigma_0$, $r(\theta)$ will be still close to $1$ even if $|\theta|$ is small.
This is interpreted that the sensitivity of measuring $I(\theta\neq0)$ increases as $\sigma_0$ decreases.
It is worth noting that the function $r(\theta)$ has a lower bound that is always positive due to the continuous spike-and-slab prior. Specifically, the lower bound is given by $r(\theta)\geq r(0)=1/[1+(1-\beta)\sigma_1/(\beta\sigma_0)]$, where $r(0)$ is a function that monotonically increases as $\sigma_0$ increases. The posterior density $\pi(\bm\theta\mid\bm n)$ is proportional to

\begin{align*}
h_1(\bm\theta):=
\exp\left(-N\log[{\bm 1}_{|I|}^T\exp(A^T\bm\theta)]+\bm n^T\bm A^T\bm\theta-\frac{\bm\theta^T\bm\theta}{2\sigma_1^2}\right)
\cdot\prod_{v^\prime<v}\left[\frac{(1-\beta)\sigma_1}{\beta\sigma_0}\exp\Big(\frac{\theta_{v^\prime v}^2}{2\sigma_1^2}-\frac{\theta_{v^\prime v}^2}{2\sigma_0^2}\Big)+1\right].
\end{align*}

In other words, $E(\bm\gamma\mid\bm n)=E_{\bm\theta\sim \pi(\bm\theta\mid\bm n)}[r(\bm\theta)]$, where applying $r(\cdot)$ to a vector means applying it element-wisely, i.e., $r(\bm\theta)=(r(\theta_{v^\prime v}):v^\prime<v)^T$.

\subsubsection{The Ising mixture model}

The posterior density of $\bm w$, $\bm\Theta$ and $\bm\Gamma$ can be written as 
\begin{align*}
\pi(\bm w,\bm\Theta,\bm\Gamma\mid\bm n)
=\prod_{k,v^\prime<v}\pi(\gamma_{v^\prime v}^{(k)}\mid\theta_{v^\prime v}^{(k)})\cdot\pi(\bm w,\bm\Theta\mid\bm n),
\end{align*}
where the posterior distribution of $\gamma_{v^\prime v}^{(k)}$ given $\theta_{v^\prime v}^{(k)}$ follows a ${\rm Bernoulli}\left(r(\theta_{v^\prime v}^{(k)})\right)$ distribution and $\pi(\bm w,\bm\Theta\mid\bm n)$ is the posterior density of $\bm w,\bm\Theta$ which is proportional to
\[
h_4(\bm w,\bm\Theta\mid\bm n):=\exp[
N\tilde{\ell}(\bm w,\bm\Theta\mid\bm n)
]
\cdot\prod_{v^\prime<v,k}\Big[\frac{(1-\beta)\sigma_1}{\beta\sigma_0}\exp\Big(\frac{[\theta_{v^\prime v}^{(k)}]^2}{2\sigma_1^2}-\frac{[\theta_{v^\prime v}^{(k)}]^2}{2\sigma_0^2}\Big)+1\Big],
\]
where
$$\tilde{\ell}(\bm w,\bm\Theta\mid\bm n)=
\ell(\bm w,\bm\Theta\mid\bm n)
+\sum_{k\in[K]}(\alpha_k-1)\log(w_k)/N
-\|\bm\Theta\|_2^2/(2N\sigma_1^2),
$$ 
and
$\ell(\bm w,\bm\Theta\mid\bm n)$ is the log-likelihood function of the Ising mixture model 
\[
\ell(\bm w,\bm\Theta\mid\bm n)=\frac{\bm n^T}{N}\log\Big(\sum_{k\in[K]}w^{(k)}\exp[A^T\bm\theta^{(k)}-\log(\bm 1_{|I|}^T\exp(A^T\bm\theta^{(k)}))]\Big).
\]

Given $\bm\Theta$, the posterior distribution of $\bm\gamma$ does not depend on the data $\bm n$. The posterior distribution of $\gamma_{v^\prime v}^{(k)}$ given $\bm n$ is a mixture of Bernoulli distributions with 
$
E(\gamma_{v^\prime v}^{(k)}\mid\bm n)=
E_{\theta_{v^\prime v}^{(k)}\sim \pi(\bm w,\bm\Theta\mid\bm n)}
[r(\theta_{v^\prime v}^{(k)})].
$
The posterior mean of $\bm w$ is
$
E(\bm w\mid\bm n)
=E_{\bm w\sim \pi(\bm w,\bm\Theta\mid\bm n)}\bm w.
$
We are specifically interested in the posterior mean of $\gamma_{v^\prime v}^{(k)}$ because its value between $0$ and $1$ reflects the magnitude of $|\theta_{v^\prime v}^{(k)}|$ through the function $|\theta|\mapsto r(|\theta|)$. As the true value of $|\theta|$ increases, the posterior mean of $\bm\Gamma$ also approaches 1, allowing us to identify the most significant non-zero interaction effects. 

\subsection{Computing posterior means}

We present importance sampling algorithms for computing the posterior means of $\bm\Gamma$ in the Ising mixture model. In the case of a single component $K=1$, we use the normal approximation as the sampling distribution. In the case of multiple components $K\geq 2$, we use the normal mixture approximation instead.

\subsubsection{The Ising model}

Recall that
$
E(\bm\gamma\mid\bm n)=E_{\bm\theta\sim \pi(\bm\theta\mid\bm n)}[r(\bm\theta)],
$
where $r(\cdot)$ is defined in \eqref{eq:r(theta)}
and
\[
\pi(\bm\theta\mid\bm n)\propto h_1(\bm\theta)
=\exp(N\tilde\ell(\bm\theta))\cdot\prod_{v^\prime<v}\left[\frac{(1-\beta)\sigma_1}{\beta\sigma_0}\exp\Big(\frac{\theta_{v^\prime v}^2}{2\sigma_1^2}-\frac{\theta_{v^\prime v}^2}{2\sigma_0^2}\Big)+1\right]
\]
with
$
\tilde\ell(\bm\theta)
:=\ell(\bm\theta)-\bm\theta^T\bm\theta/[2N\sigma_1^2]
$ and
$\ell(\bm\theta)
:=-\log[{\bm 1}_{|I|}^T\exp(A^{T}{\bm\theta})]+\bm n^T\bm A^T\bm\theta/N
$.
Note that $\ell(\bm\theta)$ is the log-likelihood function of Ising models and $\tilde{\ell}(\bm\theta)$ is its regularized version. Because the function
$$(y_1,\ldots,y_n)\mapsto \log[\exp(y_1)+\ldots+\exp(y_n)],$$  is convex, the function $\bm\theta\mapsto\tilde{\ell}(\bm\theta)$ is strictly concave, thus it has a unique maximum at the point $\tilde{\bm\theta}:=\argmax_{\bm\theta}\tilde\ell(\bm\theta)$. It follows from the Taylor series of the regularized log-likelihood $\tilde\ell$ that
\[
\tilde\ell(\bm\theta)\approx\tilde\ell(\tilde{\bm\theta})-\frac{1}{2}(\bm\theta-\tilde{\bm\theta})^T
\widetilde{\Sigma}^{-1}
(\bm\theta-\tilde{\bm\theta}),
\]
where
$\widetilde{\Sigma}$ is the inverse of Hessian matrix of $\bm\theta\mapsto-\tilde\ell(\bm\theta)$ at $\tilde{\bm\theta}$. Thus $\exp(N\tilde\ell(\bm\theta))$ can be approximated by a Normal density with mean $\tilde{\bm\theta}$ and covariance $\widetilde{\Sigma}/N$ (up to multipling a constant). After some algebra it follows that

\[
\widetilde{\Sigma}=\left[A\left\{\frac{{\rm diag}(\exp(A^T\bm\theta))}{\bm 1_{|I|}^T\exp(A^T\bm\theta)}-\frac{\exp(A^T\bm\theta)\exp(A^T\bm\theta)^T}{[\bm 1_{|I|}^T\exp(A^T\bm\theta)]^2}\right\}A^T+\frac{\bm I_{d(d+1)/2}}{N\sigma_1^2}\right]^{-1},
\]

\noindent where $\rm{diag}(\exp(A^T\bm\theta))$ represents a diagonal matrix with diagonal $\exp(A^T\bm\theta)$. Let $h_2(\bm\theta)$ be the density function of $N(\tilde{\bm\theta},\widetilde{\Sigma}/N)$. We obtain that the ratio $h_1(\bm\theta)/h_2(\bm\theta)$ is proportional to
\begin{align*}h_3(\bm\theta)
:=\exp\left(N\tilde{\ell}(\bm\theta)-N\tilde{\ell}(\tilde{\bm\theta})+\frac{N}{2}(\bm\theta-\tilde{\bm\theta})^T\widetilde{\Sigma}^{-1}(\bm\theta-\tilde{\bm\theta})\right)\cdot
\prod_{v^\prime<v}\left[\frac{(1-\beta)\sigma_1}{\beta\sigma_0}\exp\Big(\frac{\theta_{v^\prime v}^2}{2\sigma_1^2}-\frac{\theta_{v^\prime v}^2}{2\sigma_0^2}\Big)+1\right].
\end{align*}

The importance sampling estimate of the posterior mean of $\bm \gamma$ is 
\[
E(\bm \gamma\mid\bm n)
=\frac{E_{\bm\theta\sim h_2}r(\bm\theta)h_3(\bm\theta)}{E_{\bm\theta\sim h_2}h_3(\bm\theta)}
\approx
\frac{\frac{1}{M}\sum_{m\in[M]}r(\bm\theta_m)h_3(\bm\theta_m)}{\frac{1}{M}\sum_{m\in[M]}h_3(\bm\theta_m)}
=\sum_{m\in[M]}r(\bm\theta_m)\frac{h_3(\bm\theta_m)}{\sum_{m\in[M]}h_3(\bm\theta_m)},
\]
where $(\bm\theta_m,m\in[M])^T$ are i.i.d. sampled from $N(\tilde{\bm\theta},\widetilde{\Sigma}/N)$.

\subsubsection{The Ising mixture model}

Given the complexity of computing the posterior mean
of $\gamma$ in the Ising  model,
one can easily see that the computation of the posterior mean
under Ising mixture model is a lot harder. 
To resolve this issue, we use the normal mixture approximation \citep[page 85-86]{gamerman2006markov} instead of the normal approximation because the log-likelihood function $\bm w,\bm\Theta\mapsto\ell(\bm w,\bm\Theta\mid\bm n)$ may have multiple modes. The normal mixture sampling distribution denoted by $h_5(\bm w,\bm\Theta)$ can be constructed as follows.

\begin{framed}
\begin{enumerate}
\item Initialize a random start value of $\bm w,\bm\Theta$ and find a local optimal point $\tilde{\bm w},\tilde{\bm\Theta}$ that minimizes $\bm w,\bm\Theta\mapsto\tilde{\ell}(\bm w,\bm\Theta\mid\bm n)$ based on this start value.
\item Repeat the last step $J$ times. These local optimal points are denoted by $\{\tilde{\bm w}_j,\tilde{\bm\Theta}_j\}_{j=1}^{J}$ and the corresponding optimal value is $\tilde{\ell}_j:=\tilde{\ell}(\tilde{\bm w}_j,\tilde{\bm\Theta}_j\mid\bm n)$. 
Let $\tilde{\Sigma}_j$ be the inverse of Hessian matrix of $\bm\Theta\mapsto-\tilde\ell(\tilde{\bm w}_j,\bm\Theta)$ at $\tilde{\bm\Theta}_j$.
\item Let $f(\bm\Theta\mid\tilde{\bm\Theta}_j,\tilde{\Sigma}_j/N)$ be the density function of $N(\tilde{\bm\Theta}_j,\tilde{\Sigma}_j/N)$ and let $f(\bm w\mid N\tilde{\bm w}_j+\bm1_K)$ be the density function of Dirichlet distribution with parameters $N\tilde{\bm w}_j+\bm1_K$.
\item Then let $h_5(\bm w,\bm\Theta)$ be $\sum_{j\in[J]}\frac{\exp{(\tilde{\ell}_j)}}{\sum_{j\in[J]}\exp{(\tilde{\ell}_j)}}f(\bm w\mid N\tilde{\bm w}_j+\bm1_K)f(\bm\Theta\mid\tilde{\bm\Theta}_j,\tilde{\Sigma}_j/N)$.
\end{enumerate}
\end{framed}
\noindent
The number of components, $J$, in the normal mixture sampling distribution can be chosen arbitrarily, but using $J=5$ or 10 is typically sufficient to capture the majority of important modes.

The first two steps identify main local maximum points and preparing for normal approximations for each of them. The sampling distributions are constructed in the third step. We choose to use the normal approximation for the main effects and interaction effects $\bm\Theta$. We note that the mode of the Dirichlet sampling distribution of the weights $\bm w\in(0,1)^K$ corresponds to the local maximum points. 
Its parameters are then scaled by $N$ to account for the sample size, which can be justified by equating the second derivative of the objective function with that of the sampling density function in the specific scenario where the number of components $K$ is 2. These sampling distributions are combined into the full sampling distribution in which their weights are given by the corresponding values of the likelihood. The components with higher likelihood receive higher weights in the full sampling distribution.

The posterior mean of $\bm \Theta$ and $\bm w$ is given by 
\begin{align*}
E(\bm\Gamma\mid\bm n)
&=E_{\bm\Theta\sim\pi(\bm w,\bm\Theta\mid\bm n)}r(\bm\Theta)
=\int \pi(\bm w,\bm\Theta\mid\bm n)r(\bm\Theta){\sf d}\bm\Theta{\sf d}\bm w
=\int \frac{\pi(\bm w,\bm\Theta\mid\bm n)}{h_5(\bm w,\bm\Theta)}h_5(\bm w,\bm\Theta)r(\bm\Theta){\sf d}\bm\Theta{\sf d}\bm w\\
&=E_{\bm w,\bm\Theta\sim h_5(\bm w,\bm\Theta)}\frac{\pi(\bm w,\bm\Theta\mid\bm n)}{h_5(\bm w,\bm\Theta)}r(\bm\Theta)
=\frac{E_{\bm w,\bm\Theta\sim h_5(\bm w,\bm\Theta)}\frac{h_4(\bm w,\bm\Theta\mid\bm n)}{h_5(\bm w,\bm\Theta)}r(\bm\Theta)}{E_{\bm\Theta,\bm w\sim h_5(\bm w,\bm\Theta)}\frac{h_4(\bm w,\bm\Theta\mid\bm n)}{h_5(\bm w,\bm\Theta)}},
\end{align*}
where $\pi(\bm\Theta,\bm w\mid\bm n)\propto h_4(\bm\Theta,\bm w\mid\bm n)$.

As the sample size $N$ increases, the regularized log-likelihood $\tilde{\ell}$ gets closer to the log-likelihood function $\ell$. Its maximum optimal point $\tilde{\bm\theta}=\argmax_{\bm\theta}\tilde{\ell}(\bm\theta)$ gets closer to that of the log-likelihood function, $\argmax_{\bm\theta}\ell(\bm\theta)$. If the Ising mixture model is identifiable, the mean of this sampling distribution converges to the true values of $\bm\Theta$ as $N$ goes to infinity. Additionally, as $N$ increases, the covariance matrix of the sampling distribution, $\tilde{\Sigma}/N$, converges to the zero matrix, indicating that the sampling distribution is becoming more concentrated at the true values of interaction effects. It is worth noting that the classical MLE usually plays an important role in the sampling algorithm in Bayesian analysis, as demonstrated in various studies \citep{dobra2010mode,fienberg2012maximum}.
For Ising mixture models the regularized log-likelihood function $\tilde{\ell}$ can also be replaced by the log-likelihood function $\ell$ in the sampling algorithm.

Since the weights $\bm w$ are between $0$ and $1$, it is not appropriate to sample them from a Normal distribution. Instead, we sample them from a Dirichlet distribution with mode $\hat{\bm w}$. The parameter of this Dirichlet distribution is set to $N\hat{\bm w}+1$ in order to reflect the increased concentration around the mode as the sample size $N$ increases.

The posterior mean of $\bm\Gamma$ remains a meaningful method for inferring associations between variables, even if the density function $\pi(\bm n\mid\bm\Gamma)$ is non-identifiable. This is due to the fact that the posterior distribution of $\bm\Gamma$ is proportional to the product of the likelihood function $\pi(\bm\Gamma\mid\bm n)$ and the prior distribution of $\bm\Gamma$, $\pi(\bm\Gamma)$. 
If $\pi(\bm n\mid\bm\Gamma)$ is non-identifiable, meaning that multiple values of $\bm\Gamma$ produce the same likelihood, the posterior mean of $\bm\Gamma$ is the weighted average of all such values as the sample size $N\to\infty$. In particular, if the prior parameter for each element $\gamma_{v^\prime v}^{(k)}$ is set to $\beta=0.5$, the posterior mean can be interpreted as a majority vote, with a value greater than $0.5$ indicating that the majority of values of $\gamma_{v^\prime v}^{(k)}$ that produce the same likelihood are $1$. Additionally, decreasing the prior parameter, e.g., $\beta<0.5$, can favor sparser association structures.

The identifiability of the density function $\pi(\bm n\mid\bm\Gamma)$ is implied by the identifiability of $\pi(\bm n\mid\bm\Theta,\bm w)$. We discuss necessary and sufficient identifiability conditions for $\pi(\bm n\mid\bm\Theta,\bm w)$ in Section~\ref{subsec:identifiability}. On the other hand, the identifiability of $\pi(\bm n\mid\bm\Gamma)$ can be partially solved by considering the rank of the observed information matrix evaluated at the MLEs for the mixture parameters, as discussed in \citet[Chapter 9.5.2]{fruhwirth2006finite}. This is a consequence of the equivalence between local identifiability and the rank of the information matrix, as established in \citep{rothenberg1971identification,catchpole1997detecting}.

We determine the Fisher information matrix of an Ising mixture model with parameters $\bm w$, $\bm\Theta$ from the log-likelihood function
$
\ell\left({\bm w,\bm\Theta} \mid \bm X\right)=\log p_{{\rm mix},\bm X}(\bm w, \bm\Theta):
$
\[
\mathcal{I}(\bm w,\bm\Theta):=-E\left[\frac{\partial^2\log p_{{\rm mix},\bm X}(\bm w, \bm\Theta)}{\partial (\bm w,\bm\Theta)^2}\right]=-\sum_{\bm i \in I} p_{{\rm mix},\bm i}(\bm w, \bm\Theta)\frac{\partial^2\log p_{{\rm mix},\bm i}(\bm w, \bm\Theta)}{\partial (\bm w,\bm\Theta)^2}.
\]
The Fisher information matrix provides a justification for the local identifiability of an Ising mixture model\--- see Section \ref{subsec:identifiability}. 

\section{Simulation experiments} \label{subsec:simulationLogLinear}

We evaluate the empirical performance of our proposed Bayesian framework for assessing the strength of association in Ising mixture models. The number of binary variables is fixed at $d=6$. 

\subsection{The Ising model}

We set the number of variables to $d=6$, the sample size to $N=10000$, and the main effects to
\[
(\theta_1,\theta_2,\theta_3,\theta_4,\theta_5,\theta_6)=(1,-1,1,-1,1,-1).
\]
For the interaction effects $(\theta_{v^\prime v}:v^\prime<v)$ we used two designs. In design A, $(\theta_{12},\theta_{13},\theta_{14},\theta_{23})=(1,-1,1,-1)$ and others are $0$. In design B, $(\theta_{12},\theta_{13},\theta_{14},\theta_{23})=(1,-0.5,0.2,-.1)$ and others are $0$. 

We chose the following combinations for the hyperparameters of the prior distributions. In Setting 1, $\sigma_0=0.1$, $\sigma_1=1$, $\beta=0.5$. In Setting 2: $\sigma_0=0.01$, $\sigma_1=1$, $\beta=0.5$. The two settings illustrate the sensitivity of the results with respect to the ratio of the two variances in the spike-and-slab prior \eqref{eq:spikeandslab}. The sampling size $M$ in the importance sampling algorithm is $10^5$.

The data consist of the six-way contingency table with counts given by $N\cdot\bm p(\bm\theta)$, where $p(\bm\theta)$ is determined as in Equation \eqref{eq:ising}. Keeping the data fixed as opposed to sampling it from ${\rm Multinomial}(N,\bm p(\bm\theta))$ allows us to evaluate he sampling error caused by the importance sampling procedure and the performance of the proposed Bayesian method as the sample size $N$ approaches infinity. The posterior mean of the association indicators, $\bm\gamma$, is reported for each combination of designs and settings in Table~\ref{table:posteriorMean}. These results are an average of 100 independent replicates of the importance sampling algorithm with $M=10^5$. 

The results in Table~\ref{table:posteriorMean} provide evidence for the effectiveness of the proposed Bayesian framework in inferring associations between variables. Under Design A, all four non-zero interaction terms have an absolute value of 1 which makes them clearly distinguishable from the other interaction terms that are set to zero. In both prior settings, the estimated posterior means of the indicators corresponding to non-zero interaction effects is $1$, while the estimated posterior means of the zero interaction effects is $0.1$ or less.

Under Design B, three of the four non-zero interaction effects have an absolute value of $0.5$ or less. Their smaller size makes them less distinguishable from the remaining interaction terms that are set to zero. In Setting 1 ($\sigma_0=.1$), the estimated posterior means of the association indicators for larger interaction effects is close to $1$. The estimated posterior means for smaller interaction effects is less than $0.5$, suggesting that these associations may be harder to identify. In Setting 2 ($\sigma_0=.01$), the estimated posterior means of the association indicators for all four non-zero interaction effects is greater than $0.5$, demonstrating the increased effectiveness of a smaller value of $\sigma_0$ in detecting small interaction effects.

The estimated posterior mean of $\bm\gamma$ is smaller in Setting 1 than in Setting 2 for the interaction effects that are set to zero. This should not be surprising, since the expectation $E(\gamma\mid\bm n)=E_{\theta\sim \pi(\bm\theta\mid\bm n)}[r(\theta)]$, where $r(\theta)$ is defined in \eqref{eq:r(theta)} is monotonically increasing with respect to $\sigma_0$ for small values of $|\theta|$. As a result, when $\sigma_0$ approaches 0, the lower bound of the posterior mean of $\gamma$, which is $r(0) = 1/(1+(1-\beta)\sigma_1/\beta\sigma_0)$, becomes smaller.
On the other hand, the sampling error is larger in Setting 2 ($\sigma_0=0.01$) compared to Setting 1 ($\sigma_0=0.1$). Under both designs the importance sampling standard errors for the posterior mean estimates of association indicators are approximately $0.0001$ in Setting 1 and $0.1$ in Setting 2. The reason relates to the sampling density function $h_2(\bm\theta)$ being closer to the objective density function $\pi(\bm\theta\mid\bm n)$ in Setting 1, resulting in a lower variance of the importance sampling method. As such, selecting the value of $\sigma_0$ involves balancing the stability of the importance sampling algorithm with the ability to detect distinguish non-zero interaction effects.

{
\renewcommand{\tabcolsep}{1.5pt}
\renewcommand{\arraystretch}{1.0}
\begin{table}[ht]
\centering
{{
\begin{tabular}{
C{.37in}C{.37in}C{.37in}C{.37in}C{.37in}C{.37in}C{.37in}
C{.37in}
C{.37in}C{.37in}C{.37in}C{.37in}C{.37in}C{.37in}C{.37in}}
\toprule
$\gamma_{12}$ & 
$\gamma_{13}$ & 
$\gamma_{14}$ &
$\gamma_{15}$ &
$\gamma_{16}$ &
$\gamma_{23}$ & 
$\gamma_{24}$ & 
$\gamma_{25}$ & 
$\gamma_{26}$ &
$\gamma_{34}$ & 
$\gamma_{35}$ & 
$\gamma_{36}$ &
$\gamma_{45}$ & 
$\gamma_{46}$ & 
$\gamma_{56}$ \\
\multicolumn{15}{c}{
Posterior mean 
under Design A and Setting 1} 
\\
\textbf{1.0} & \textbf{1.0} & \textbf{1.0} & .10 & .10 & \textbf{1.0} & .10 & .10  & .10 & .10 & .10 & .10 & .10 & .10 & .10 \\
\multicolumn{15}{c}{
Posterior mean 
under Design B and Setting 1} 
\\
\textbf{1.0} & \textbf{1.0} & .34 & .10 & .10 & .14 & .10 & .10 & .10 & .10 & .10 & .10 & .10 & .10 & .10 \\
\multicolumn{15}{c}{
Posterior mean 
under Design A and Setting 2} 
\\
\textbf{1.0} & \textbf{1.0} & \textbf{1.0} & .08 & .10 & \textbf{1.0} & .08 & .07 & .08 & .06 & .09 & .07 & .07 & .06 & .09\\
\multicolumn{15}{c}{
Posterior mean 
under Design B and Setting 2} 
\\
\textbf{1.0} & \textbf{1.0} & \textbf{1.0} & .10 & .10 & \textbf{.68} & .07 & .08 & .07 & .08 & .08 & .09 & .08 & .10 & .08\\
\bottomrule 
\end{tabular}}}
\caption{
Estimated posterior mean of $\bm\gamma$ in Ising models under Design A, B, and Setting 1 and 2. Under both designs, the importance sampling standard errors are approximately $0.0001$ for Setting 1 and $0.1$ for Setting 2.
}
\label{table:posteriorMean}
\end{table}
}

\subsection{The Ising mixture model with two components}

We set the number of variables to $d=6$, the sample size $N=10000$, the weight for the first component to $w^{(1)}=0.4$, and the main effects to $\bm\theta^{(1)}=\bm\theta^{(2)}=(1,-1,1,-1,1,-1)$. The spike-and-slab prior parameters are set to $\sigma_0=0.1$, $\sigma_1=1$, $\beta=.5$. This is Setting 1 in the previous simulation experiment. The sampling size $M$ in the importance sampling algorithm is $10^5$. The number of components $J$ in the normal mixture sampling distribution for Bayesian Ising mixture models is $5$.

The data consist of the six-way contingency table with counts given by $\bm n=N\cdot\bm p_{\rm mix}(\bm w,\bm\Theta)$ where $p_{\rm mix}(\bm w,\bm\Theta)$ is defined in Equation \eqref{eq:isingmixture}. We used two designs for the interaction effects. In design C, $(\theta^{(1)}_{12},\theta^{(1)}_{13},\theta^{(2)}_{46},\theta^{(2)}_{56})=(1,-1,1,-1)$ and others are $0$. In design D, $(\theta^{(1)}_{12},\theta^{(1)}_{13},\theta^{(1)}_{23},\theta^{(2)}_{14},\theta^{(2)}_{15})=(1,-1,1,1,-1)$ and others are $0$. Under both designs, we fit an Ising model as well as an Ising mixture model with two components. Table~\ref{table:posteriorMean3} presents the estimated posterior means of the association indicators for both models. These results are an average of 100 independent replicates of the importance sampling algorithm with $M=10^5$. The importance sampling standard error is about $0.05$ for both designs which is an indication of the stability of the importance sampling algorithm. 
As an illustration of computation time on average, the Bayesian Ising model required only 0.37 seconds, while the Bayesian two-component Ising mixture model took 5.14 minutes to complete the importance sampling algorithm under Design D and Setting 1. Both experiments were conducted on a laptop with a 1.8 GHz Intel Core i5 processor and 8 GB of memory.

Under both designs, the Ising model identifies the non-zero interaction effects from both components of the mixture. However, under Design D, it incorrectly identifies two additional interaction effects that are actually zero in both mixture components. On the other hand, the Ising mixture model with two components identifies all non-zero and all zero interaction effects in both components based on a cutoff of $0.5$ under both designs. The estimated posterior mean of the weight $w^{(1)}$ is $.40$ under Design C and $0.41$ under Design D, both estimates very close to the true value of $0.4$.

This simulation setting shows that, if an Ising mixture model with one component is fit when the data corresponds with a mixture model with two components, the inferred association structure might include pairwise effects that do not exist. This result is in a sense not surprising given that the first component has non-zero interaction effects between variables 1 and 2, variables 1 3, and variables 2 and 3, while the second component has non-zero interaction effects between variables 1 4, and variables 1 and 5. Due to this configuration, the Ising model might show non-zero interaction effects between variables 2 and 4, and variables 2 and 5. Nevertheless, the estimated posterior mean of the association indicators for these variables is lower compared to the truly non-zero associations. 

{
\renewcommand{\tabcolsep}{1.5pt}
\renewcommand{\arraystretch}{1.0}
\begin{table}[ht]
\centering
{{
\begin{tabular}{
C{1.0in}
C{.3in}C{.3in}C{.3in}C{.3in}C{.3in}C{.3in}C{.3in}
C{.3in}
C{.3in}C{.3in}C{.3in}C{.3in}C{.3in}C{.3in}C{.3in}}
\toprule
Model&
$\gamma_{12}$ & 
$\gamma_{13}$ & 
$\gamma_{14}$ &
$\gamma_{15}$ &
$\gamma_{16}$ &
$\gamma_{23}$ & 
$\gamma_{24}$ & 
$\gamma_{25}$ & 
$\gamma_{26}$ &
$\gamma_{34}$ & 
$\gamma_{35}$ & 
$\gamma_{36}$ &
$\gamma_{45}$ & 
$\gamma_{46}$ & 
$\gamma_{56}$ \\
\hline\\
\multicolumn{16}{c}{
Posterior mean 
under Design C} 
\\
Ising model&
\textbf{.99} & \textbf{.98} & .10 & .10 & .11 & .30 & .11 & .11 & .14 & .11 & .11 & .14 & .11 & \textbf{1.0} & \textbf{1.0}\\
\hline
Ising mixture, Component 1&
\textbf{1.0} & \textbf{1.0} & .17 & .16 & .20 & .19 & .16 & .15  & .17 & .18 & .17 & .18 & .16 & .19 & .19 \\
\hline
Ising mixture, Component 2&
.17 & .14 & .14 & .13 & .13 & .16 & .17 & .16  & .15 & .14 & .13 & .13 & .15 & \textbf{1.0} & \textbf{1.0} \\
\hline\\
\multicolumn{16}{c}{
Posterior mean 
under Design D} 
\\
Ising model &
\textbf{.98} & \textbf{.94} & \textbf{1.0} & \textbf{1.0} & .10 & \textbf{.88} & \textbf{.69} & \textbf{.67} & .10 & .13 & .12 & .13 & .22 & .22 & .22\\
\hline
Ising mixture, Component 1&
\textbf{.99} & \textbf{.99} & .19 & .17 & .14 & \textbf{.99} & .17 & .15 & .15 & .17 & .16 & .14 & .17 & .18 & .15\\
\hline
Ising mixture, Component 2&
.19 & .14 & \textbf{.99} & \textbf{.99} & .12 & .17 & .15 & .17 & .15 & .12 & .12 & .12 & .12 & .12 & .12\\
\bottomrule 
\end{tabular}}}
\caption{
Estimated posterior mean of $\bm\gamma$ in the Ising model and of $\bm\gamma^{(1)},\bm\gamma^{(2)}$ in the Ising mixture model with two components under Design C and Design D.
}
\label{table:posteriorMean3}
\end{table}
}

\section{Real data applications} \label{subsec:applicationlogLinear}

We examine the fit of the Ising model and of the Ising mixture model with two components for two eight-way binary contingency tables. The first example focuses on the Rochdale data \--- a dataset that has been analyzed numerous times in the existent literature. The pairwise interactions of the Rochdale data are considered to be well understood. The second dataset comes from a larger dataset, and has not been previously analyzed in this form.  We chose it  because of its much larger sample size leads to significantly larger counts in some of the cells compared to the largest counts in the Rochdale data. The presence of the larger counts will test the ability of the Ising mixture models to adequately capture the imbalance between the magnitude of the largest and the smallest counts in sparse contingency tables.

\subsection{The Rochdale data}
The Rochdale data \citep{whittaker1990graphical} was collected to determine the relationships among factors affecting women's economic activity. 
It includes eight binary variables: 1) wife's economic activity (no, yes), 2) age of wife $>38$ (no, yes), 3) husband's employment status (no, yes), 4) presence of children $\leq 4$ years old (no, yes), 5) wife's education level, high-school+ (no, yes), 6) husband's education level, high-school+ (no, yes), 7) Asian origin (no, yes), and 8) presence of other working household members (no, yes). With a sample size of 665, the resulting $2^{8}$ contingency table, shown in Table~\ref{tab:rochadale_data}, is sparse, with 165 cells having 0 counts, 217 cells having small positive counts less than 3, and several cells with counts larger than 30 or even 50.

\begin{table}[ht]
\centering
\begin{tabular}{rrrrrrrrrrrrrrrr}
\hline 5 & 0 & 2 & 1 & 5 & 1 & 0 & 0 & 4 & 1 & 0 & 0 & 6 & 0 & 2 & 0 \\
8 & 0 & 11 & 0 & 13 & 0 & 1 & 0 & 3 & 0 & 1 & 0 & 26 & 0 & 1 & 0 \\
5 & 0 & 2 & 0 & 0 & 0 & 0 & 0 & 0 & 0 & 0 & 0 & 0 & 0 & 1 & 0 \\
4 & 0 & 8 & 2 & 6 & 0 & 1 & 0 & 1 & 0 & 1 & 0 & 0 & 0 & 1 & 0 \\
17 & 10 & 1 & 1 & 16 & 7 & 0 & 0 & 0 & 2 & 0 & 0 & 10 & 6 & 0 & 0 \\
1 & 0 & 2 & 0 & 0 & 0 & 0 & 0 & 1 & 0 & 0 & 0 & 0 & 0 & 0 & 0 \\
4 & 7 & 3 & 1 & 1 & 1 & 2 & 0 & 1 & 0 & 0 & 0 & 1 & 0 & 0 & 0 \\
0 & 0 & 3 & 0 & 0 & 0 & 0 & 0 & 0 & 0 & 0 & 0 & 0 & 0 & 0 & 0 \\
18 & 3 & 2 & 0 & 23 & 4 & 0 & 0 & 22 & 2 & 0 & 0 & 57 & 3 & 0 & 0 \\
5 & 1 & 0 & 0 & 11 & 0 & 1 & 0 & 11 & 0 & 0 & 0 & 29 & 2 & 1 & 1 \\
3 & 0 & 0 & 0 & 4 & 0 & 0 & 0 & 1 & 0 & 0 & 0 & 0 & 0 & 0 & 0 \\
1 & 1 & 0 & 0 & 0 & 0 & 0 & 0 & 0 & 0 & 0 & 0 & 0 & 0 & 0 & 0 \\
41 & 25 & 0 & 1 & 37 & 26 & 0 & 0 & 15 & 10 & 0 & 0 & 43 & 22 & 0 & 0 \\
0 & 0 & 0 & 0 & 2 & 0 & 0 & 0 & 0 & 0 & 0 & 0 & 3 & 0 & 0 & 0 \\
2 & 4 & 0 & 0 & 2 & 1 & 0 & 0 & 0 & 1 & 0 & 0 & 2 & 1 & 0 & 0 \\
0 & 0 & 0 & 0 & 0 & 0 & 0 & 0 & 0 & 0 & 0 & 0 & 0 & 0 & 0 & 0 \\
\hline
\end{tabular}
\caption{Rochdale data from \citet{whittaker1990graphical}. The cells counts appear row by row in lexicographical order with the levels of variable $8$ varying fastest and the levels of variable $1$ varying slowest.}
\label{tab:rochadale_data}
\end{table}

The estimated posterior means of the association indicators $\bm\gamma$ are presented in Table~\ref{table:posteriorMeanRochdale} based on an Ising model with spike-and-slab prior with $\sigma_0=0.1$, $\sigma_1=1$ and $\beta=0.5$. In what follows we consider an interaction effect between variables $v^\prime$ and $v$ to be significant if $E(\gamma_{v^\prime v}\mid\bm n) > 0.5$. In Figure~\ref{fig:RochdaleIsing} we compare the set of non-zero interaction effects we identified with those of \citet{whittaker1990graphical}. We find that all the $14$ significant pairwise associations found by \citet{whittaker1990graphical} are also found by our Ising model. However, the Ising model determined  two additional interactions: one interaction between variables $5$ (wife's education level) and $7$ (Asian origin) with a posterior mean of the corresponding association indicator of $0.86$, and the interaction between variables $2$ (age of wife $>38$) and $7$ (Asian origin) with a posterior mean of $0.65$. The posterior means of these two extra interactions are much smaller then the posterior means of the $14$ interactions that were also determined by \citet{whittaker1990graphical}. These two extra associations seem to be reasonable, and can be attributed to the wave of Asian immigration, particularly of Asian women, in the last century \citep{kim1977asian,piper2004wife}.

We also applied our Bayesian Ising mixture model with two components to the Rochdale data \--- see Table~\ref{table:posteriorMeanRochdale} and Figure~\ref{fig:RochdaleIsingMixture}. 
We fitted the mixture model with invariant main effects across components (Assumption \ref{assumption:sameMainEffectsForGeneralK}). The number of components in the normal mixture sampling distribution is $J=5$. The estimated posterior mean of the weight of the first component is $E(w^{(1)}\mid\bm n)=0.14$. The 16 significant associations found by the Bayesian Ising model are also idenfied in both components of the Ising mixture model. However, each of the two components involve additional significant associations. We note that the estimated posterior means of all the association indicators are $1$. Goodness-of-fit tests for the maximum likelihood estimators show that both the Ising model \eqref{eq:ising} and the Ising mixture model \eqref{eq:isingmixture} fit the data well with p-values of $0.42$ and $1$, respectively. The likelihood ratio test shows that the two-component Ising mixture model fits the data significantly better than the Ising model with a p-value $<0.001$.

The left panel of Table~\ref{table:ComparingRochdale} shows the cells containing the 10 largest observed counts in the Rochdale data together with their expected cell counts in the Ising model and the Ising mixture model. We see that both models are able to capture the largest counts reasonably well.

{
\renewcommand{\tabcolsep}{1.5pt}
\renewcommand{\arraystretch}{1.0}
\begin{table}[ht]
\centering
{{
\begin{tabular}{C{.61in}
C{.4in}C{.6in}C{.6in}C{.4in}C{.61in}
C{.4in}C{.6in}C{.6in}}
\toprule
    Cell &  Observed & Ising models &  Ising mixtures & & Cell &  Observed & Ising models &  Ising mixtures \\
\cmidrule(r){1-4} \cmidrule(r){6-9}
10001100 &        57 &       56.78 &           58.63 & & 00000000 &      4419 &  4181.60 &         4320.54 \\
11001100 &        43 &       44.61 &           42.89 & & 00010000 &      2063 &   2087.60 &         2134.16 \\
11000000 &        41 &       36.40 &           36.99 & & 00110000 &      1189 &  1324.38 &         1175.31 \\
11000100 &        37 &       38.81 &           37.65 & & 11111111 &      1056 &    1035.05 &         1055.71 \\
10011100 &        29 &   33.29 &           29.02 & & 00111111 &       764 &      702.14 &          752.47 \\
00011100 &        26 &       20.37 &           21.84 & & 00110100 &       667 &    607.96 &          658.85 \\
11000101 &        26 &      23.69 &           23.33 & & 00110101 &       654 &    702.29 &          657.78 \\
11000001 &        25 &      28.13 &           29.30 & & 00110001 &       601 &   571.91 &          597.43 \\
10000100 &        23 &       22.70 &           23.25 & & 00111101 &       549 &   577.10 &          561.45 \\
11001101 &        22 &       22.85 &           21.43 & & 00010100 &       529 &   565.21 &          518.92 \\
\bottomrule
\end{tabular}
}}
\caption{
Expected cell counts for the top 10 largest counts cells for the Rochdale data (left panel) and the NLTCS data (right panel). Cells are identified through their sequence of level indicators with 0 as no and 1 as yes.
}
\label{table:ComparingRochdale}
\end{table}
}

{
\renewcommand{\tabcolsep}{1.5pt}
\renewcommand{\arraystretch}{1.0}
\begin{table}[ht]
\centering
{{
\begin{tabular}{
C{1.0in}
C{.3in}C{.3in}C{.3in}C{.3in}C{.3in}C{.3in}
C{.3in}
C{.3in}C{.3in}C{.3in}C{.3in}C{.3in}C{.3in}C{.3in}}
\toprule
Model&
$\gamma_{12}$ & 
$\gamma_{13}$ & 
$\gamma_{14}$ &
$\gamma_{15}$ &
$\gamma_{16}$ &
$\gamma_{17}$ & 
$\gamma_{18}$ & 
$\gamma_{23}$ & 
$\gamma_{24}$ &
$\gamma_{25}$ & 
$\gamma_{26}$ & 
$\gamma_{27}$ &
$\gamma_{28}$ & 
$\gamma_{34}$ \\
\hline\\
Ising model&
.23 & \textbf{1.0} & \textbf{1.0} & \textbf{.96} & .22 & \textbf{1.0} & .21 & .29 & \textbf{1.0} & \textbf{1.0} & .18 & \textbf{.65} & \textbf{1.0} & .25\\
\hline
Ising mixture, Component 1&
\textbf{1.0} & \textbf{1.0} & \textbf{1.0} & \textbf{1.0} & \textbf{1.0} & \textbf{1.0} & \textbf{1.0} & \textbf{1.0} & \textbf{1.0} & \textbf{1.0} & \textbf{1.0} & \textbf{1.0} & \textbf{1.0} & \textbf{1.0}\\
\hline
Ising mixture, Component 2&
.11 & \textbf{1.0} & \textbf{1.0} & \textbf{1.0} & \textbf{.92} & \textbf{1.0} & \textbf{1.0} & \textbf{1.0} & \textbf{1.0} & \textbf{1.0} & .29 & \textbf{1.0} & \textbf{1.0} & \textbf{1.0}\\ \bottomrule
\\
\\ \toprule
Model&
$\gamma_{35}$ & 
$\gamma_{36}$ & 
$\gamma_{37}$ &
$\gamma_{38}$ &
$\gamma_{45}$ &
$\gamma_{46}$ & 
$\gamma_{47}$ & 
$\gamma_{48}$ & 
$\gamma_{56}$ &
$\gamma_{57}$ & 
$\gamma_{58}$ & 
$\gamma_{67}$ &
$\gamma_{68}$ & 
$\gamma_{78}$ \\
\hline\\
Ising model &
\textbf{1.0} & \textbf{.95} & \textbf{.98} & .28 & .30 & .46 & \textbf{.99} & \textbf{.99} & \textbf{1.0} & \textbf{.86} & .37 & \textbf{1.0} & .44 & .37\\
\hline
Ising mixture, Component 1&
\textbf{1.0} & \textbf{1.0} & \textbf{1.0} & \textbf{1.0} & \textbf{1.0} & \textbf{1.0} & \textbf{1.0} & \textbf{1.0} & \textbf{1.0} & \textbf{1.0} & \textbf{1.0} & \textbf{1.0} & \textbf{1.0} & \textbf{1.0}\\
\hline
Ising mixture, Component 2&
\textbf{1.0} & \textbf{1.0} & \textbf{1.0} & \textbf{1.0} & .24 & \textbf{1.0} & .34 & \textbf{1.0} & \textbf{1.0} & .29 & \textbf{1.0} & \textbf{1.0} & \textbf{1.0} & .39\\
\bottomrule 
\end{tabular}}}
\caption{
Estimated posterior means of the association indicators in the Ising model and the Ising mixture model with two components for the Rochdale data.
}
\label{table:posteriorMeanRochdale}
\end{table}
}

\begin{figure}[ht]
\includegraphics[width=.5\textwidth]{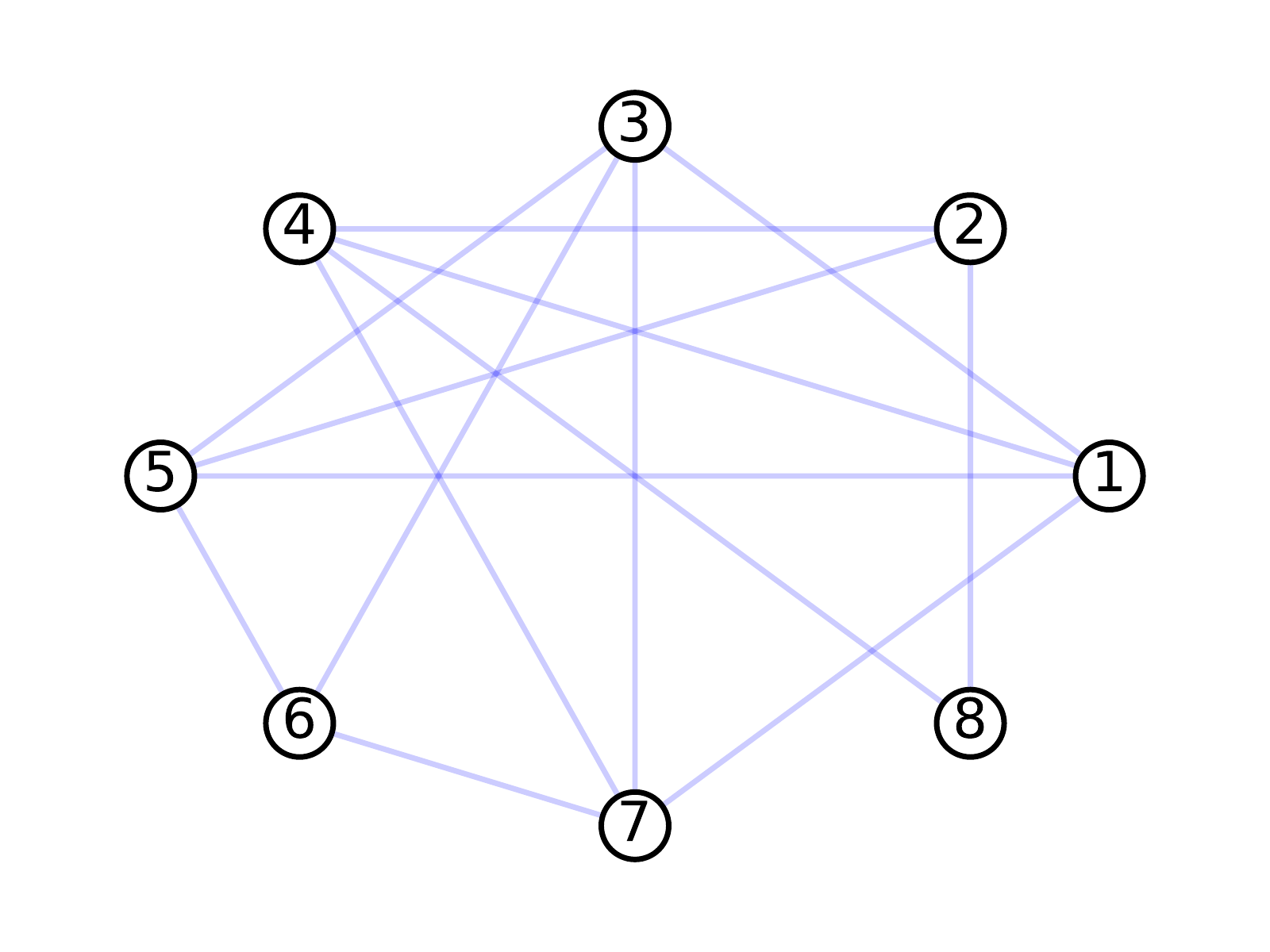}
\includegraphics[width=.5\textwidth]{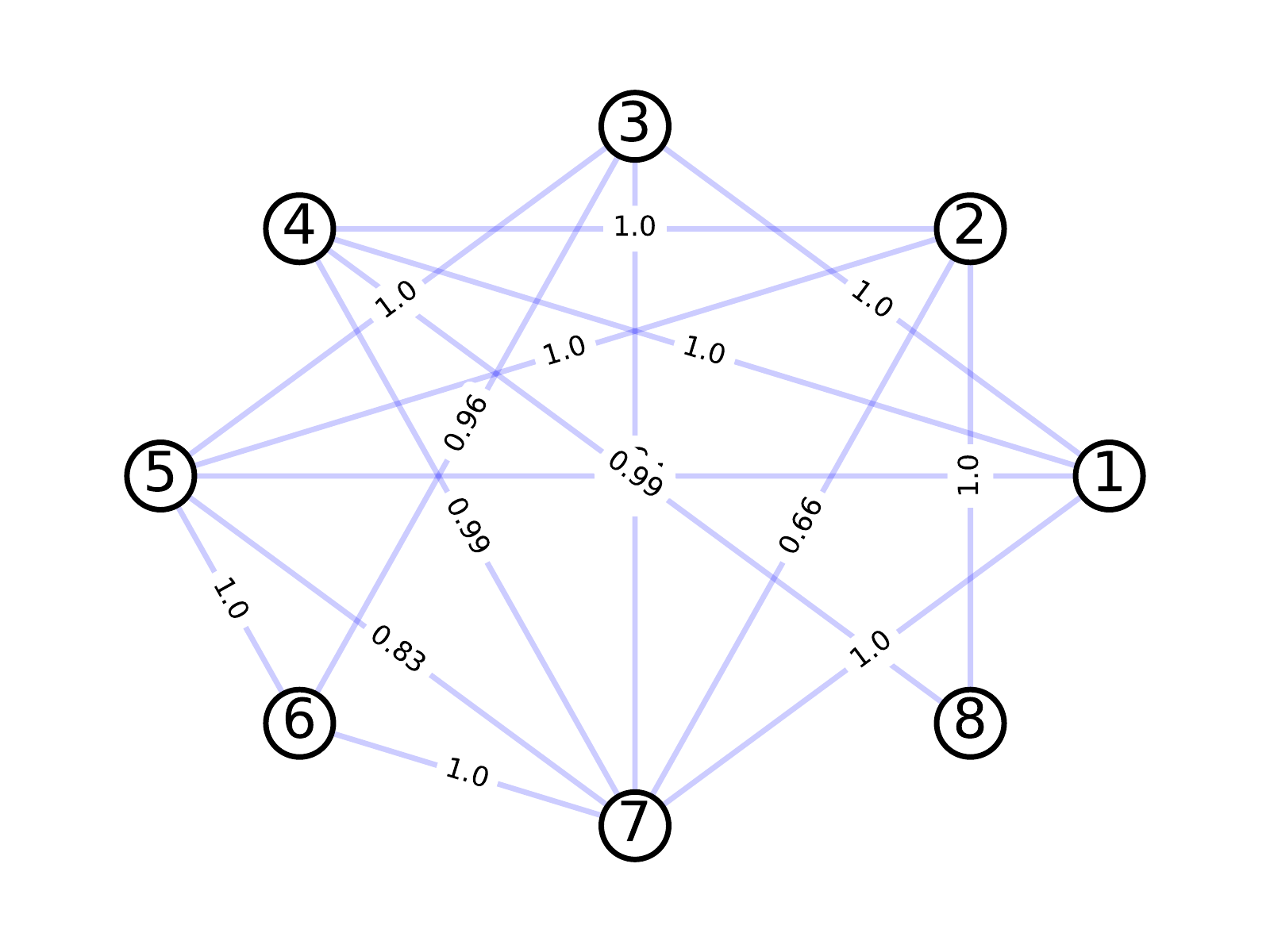}
\caption{Significant pairwise associations determined by \citet{whittaker1990graphical} (left panel) and the Bayesian Ising model we proposed (right panel). Each association is shown as an edge between vertices associated with variables it involves. The labels of the edges indicate the estimated posterior means of their indicators.
}
\label{fig:RochdaleIsing}
\includegraphics[width=.5\textwidth]{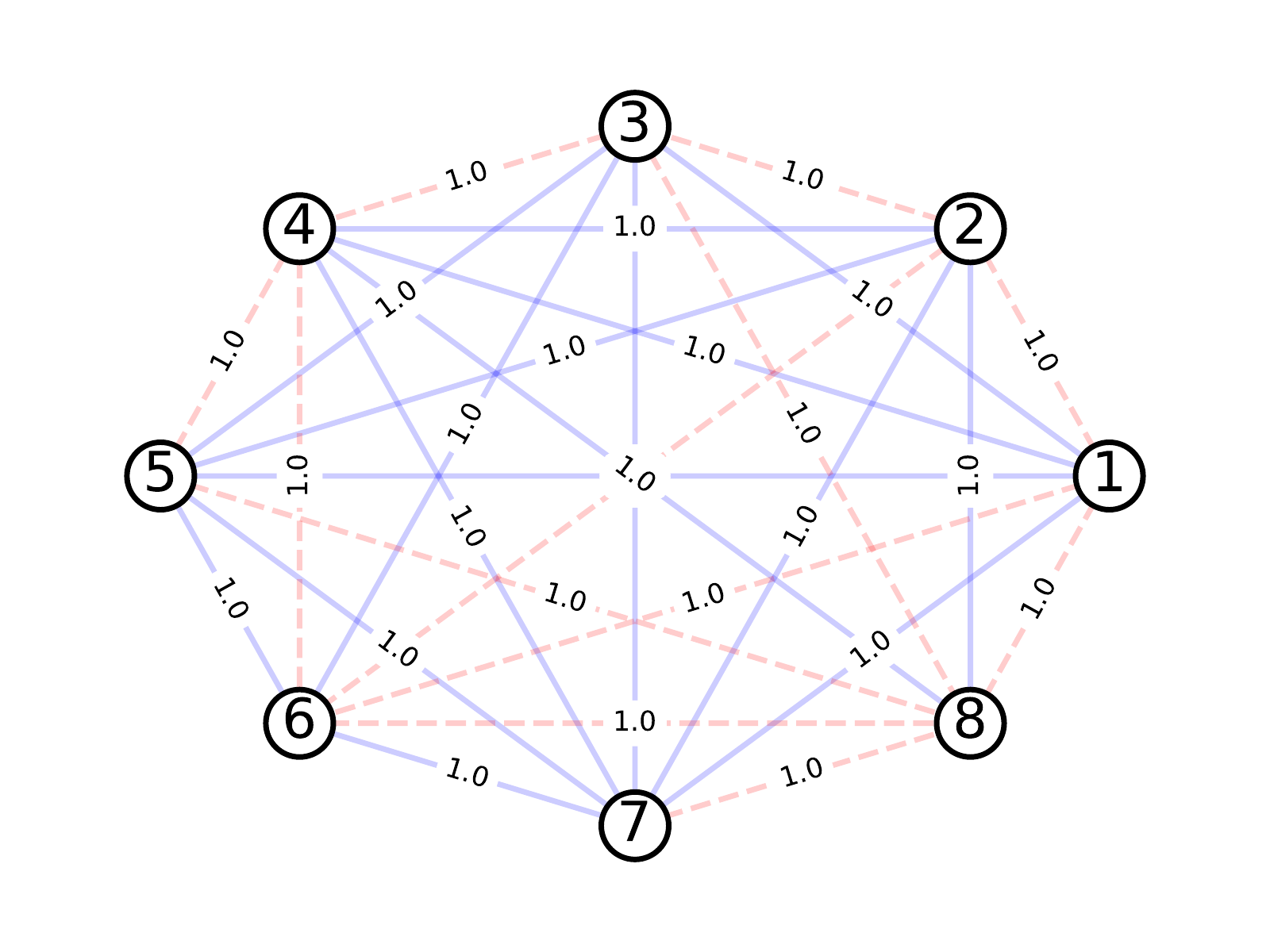}
\includegraphics[width=.5\textwidth]{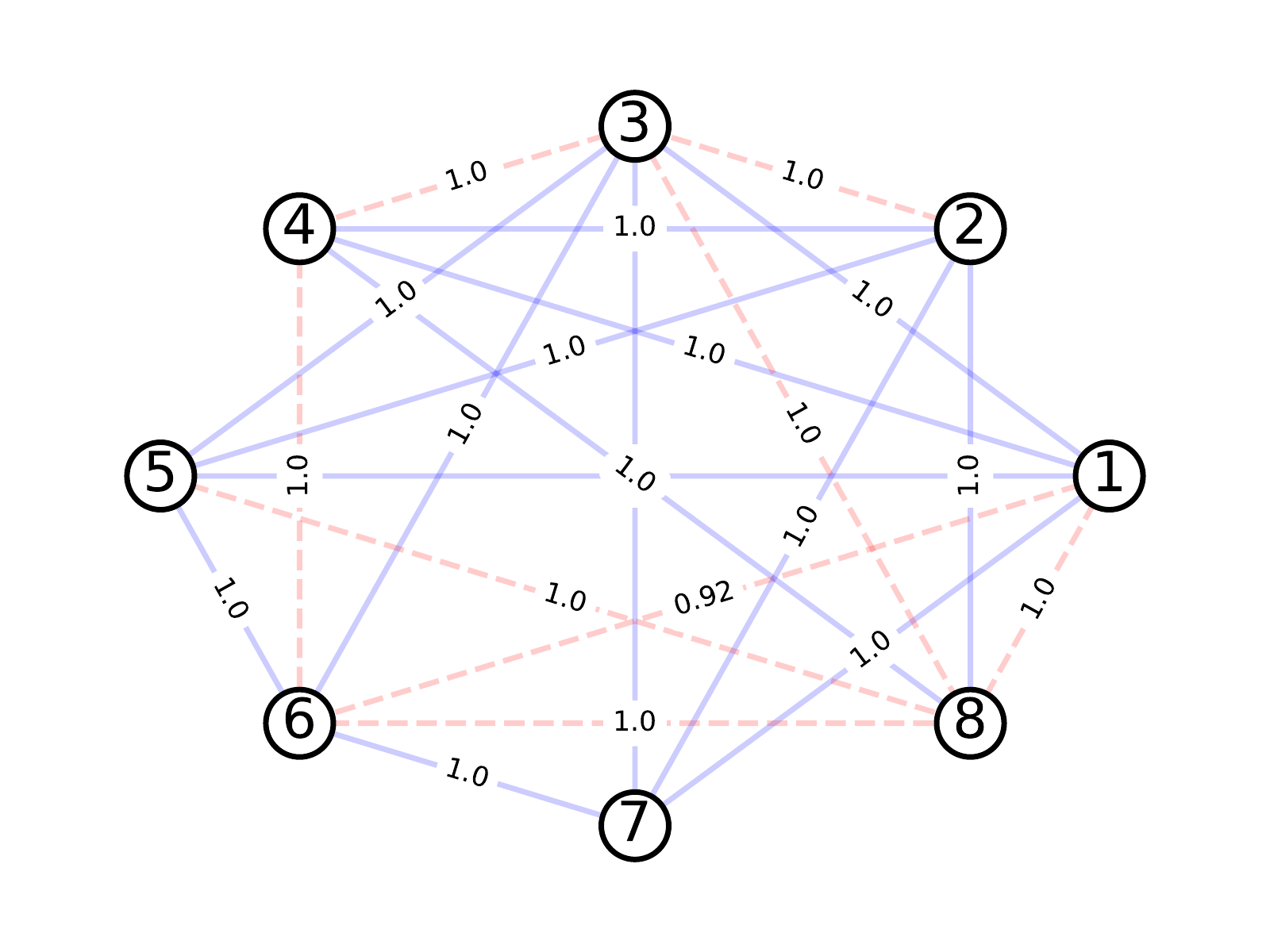}
\caption{
Significant pairwise associations determined by the Bayesian Ising mixture model. The 28 associations in the first component are presented in the left panel, while the 22 associations in the right component are shown in the right panel. Each association is shown as an edge between vertices associated with variables it involves. The labels of the edges indicate the estimated posterior means of their indicators.}
\label{fig:RochdaleIsingMixture}
\end{figure}

\subsection{The NLTCS data}
We analyze a dataset extracted from the National Long Term Care Survey (NLTCS) created by the Center of Demographic Studies at Duke University \citep{manton-et-1993}. 
It includes eight binary variables that measure functional disability in daily living activities: 1) eating, 2) getting around inside, 3) dressing, 4) cooking, 5) grocery shopping, 6) getting about outside, 7) traveling, and 8) managing money. Each measure classifies study participants as healthy or disabled. The data comprise observations of  elderly individuals aged 65 and above, pooled across four survey waves from 1982, 1984, 1989, and 1994. With a sample size of 21574, the resulting $2^{8}$ contingency table, Table~\ref{tab:NLTCSdataChapter4} is sparse, with $17.1\%$ cells having 0 counts, $46.9\%$ cells having small counts no larger than $5$, and largest $1.6\%$ of the cell counts accounting for $40.5\%$ of the observations.

\begin{table}[ht]
\centering
\begin{tabular}{rrrrrrrrrrrrrrrr}
\hline 
4419 & 97 & 67 & 472 & 2063 & 55 & 335 & 44 & 313 & 18 & 33 & 76 & 1 & 5 & 2 & 6 \\
119 & 115 & 1 & 16 & 0 & 4 & 1189 & 17 & 112 & 6 & 130 & 64 & 529 & 52 & 453 & 56 \\
2 & 22 & 13 & 116 & 10 & 67 & 47 & 0 & 2 & 0 & 1 & 92 & 0 & 4 & 0 & 4 \\
1 & 12 & 5 & 19 & 1 & 0 & 0 & 3 & 1 & 0 & 354 & 2 & 27 & 4 & 16 & 5 \\
55 & 3 & 24 & 1 & 0 & 0 & 1 & 7 & 1 & 60 & 667 & 29 & 601 & 14 & 1 & 16 \\
3 & 55 & 8 & 85 & 7 & 65 & 69 & 400 & 24 & 5 & 62 & 2 & 10 & 164 & 0 & 8 \\
2 & 6 & 3 & 15 & 3 & 5 & 0 & 0 & 0 & 0 & 0 & 0 & 3 & 32 & 4 & 41 \\
1 & 0 & 1 & 0 & 1 & 0 & 4 & 1 & 3 & 3 & 9 & 0 & 0 & 0 & 0 & 0 \\
14 & 226 & 11 & 140 & 5 & 0 & 2 & 2 & 10 & 0 & 7 & 3 & 3 & 11 & 31 & 3 \\
0 & 4 & 0 & 2 & 125 & 8 & 134 & 81 & 654 & 34 & 0 & 34 & 1 & 5 & 25 & 215 \\
8 & 80 & 30 & 5 & 105 & 19 & 50 & 1 & 1 & 0 & 2 & 3 & 0 & 3 & 0 & 3 \\
13 & 9 & 4 & 1 & 0 & 0 & 4 & 7 & 1 & 6 & 6 & 54 & 3 & 0 & 1 & 0 \\
0 & 1 & 6 & 0 & 6 & 1 & 42 & 3 & 28 & 48 & 207 & 12 & 0 & 5 & 0 & 2 \\
4 & 34 & 1 & 13 & 6 & 38 & 549 & 19 & 180 & 21 & 196 & 27 & 2 & 14 & 72 & 88 \\
8 & 3 & 0 & 0 & 2 & 8 & 11 & 3 & 15 & 9 & 5 & 19 & 3 & 26 & 0 & 28 \\
29 & 158 & 10 & 89 & 5 & 66 & 764 & 66 & 86 & 8 & 175 & 7 & 151 & 131 & 516 & 1056 \\
\hline
\end{tabular}
\caption{The NLTCS data. This 16 by 16 tables shows all the possible combination of the $8$ binary variables. The cells counts appear row by row in lexicographical order with Variable $8$ varying fastest and Variable $1$ varying slowest.}
\label{tab:NLTCSdataChapter4}
\end{table}

We employ our Bayesian framework to fit an Ising model and an Ising mixture model with two components. The prior specification, assumptions related to the invariance of main effects and the number of components in the mixture sampling distributions were the same as the ones used for the Rochdale data. The pairwise associations that have an estimated posterior mean of their indicators above $0.5$ are shown as graphs in the right panel of Figure~\ref{fig:NLTCSbeta50IsingEdges}. There are 21 significant interaction effects identified in the Ising model. The forward stepwise function from the R package gRim \citep{hojsgaard2012graphical} identifies 20 of these 21 pairwise interactions \--- see the left panel of Figure~\ref{fig:NLTCSbeta50IsingEdges}. The additional interaction identified by our Bayesian framework involves variables 1 and 4, and has the smallest estimated posterior mean of $0.85$ among the 21 associations. In the Ising mixture model, there are 17 significant interaction effects in the first component and 20 significant interaction effects in the second component \--- see Figure~\ref{fig:NLTCSbeta50IsingMixture}. The estimated posterior mean of the weight of the first component is $0.4$. The patterns of significant interaction effects in both components of the Ising mixture model are sparser than the pattern inferred in the Ising model. One example of a key difference between the inferred association patterns relates to variables $1$ and $4$. The estimated posterior mean of the association indicator is $0.85$ in the Ising model, while it is less than $0.5$ in the first component of the Ising mixture model and it is equal with $1$ in the second component. As such, this pairwise association is a combination of a weaker effect in one component and a very strong effect in the second component. Similar patterns with varying strength between components involve variables $1$ and $3$ and variables $4$ and $5$.

Goodness-of-fit tests for the maximum likelihood estimators show that the Ising model \eqref{eq:ising} does not fit the data well (p-value $<0.00001$), while the Ising mixture model \eqref{eq:isingmixture} with two components fits the data well (p-value $0.21$). The right panel of Table~\ref{table:ComparingRochdale} shows the cells containing the 10 largest observed counts observed in the NLTCS data together with their expected cell counts in the Ising model and the Ising mixture model. We see that the Ising mixture model seems to capture the size of the largest cell counts more faithfully than the Ising model.

{
\renewcommand{\tabcolsep}{1.5pt}
\renewcommand{\arraystretch}{1.0}
\begin{table}[ht]
\centering
{{
\begin{tabular}{
C{1.0in}
C{.3in}C{.3in}C{.3in}C{.3in}C{.3in}C{.3in}
C{.3in}
C{.3in}C{.3in}C{.3in}C{.3in}C{.3in}C{.3in}C{.3in}}
\toprule
Model&
$\gamma_{12}$ & 
$\gamma_{13}$ & 
$\gamma_{14}$ &
$\gamma_{15}$ &
$\gamma_{16}$ &
$\gamma_{17}$ & 
$\gamma_{18}$ & 
$\gamma_{23}$ & 
$\gamma_{24}$ &
$\gamma_{25}$ & 
$\gamma_{26}$ & 
$\gamma_{27}$ &
$\gamma_{28}$ & 
$\gamma_{34}$ \\
\hline\\
Ising model&
\textbf{1.0} & \textbf{.90} & \textbf{.85} & .32 & .12 & \textbf{1.0} & .11 & \textbf{1.0} & \textbf{1.0} & .29 & \textbf{.98} & \textbf{1.0} & .28 & \textbf{1.0}\\
\hline
Ising mixture, Component 1&
\textbf{1.0} & \textbf{1.0} & .16 & \textbf{.70} & .35 & .27 & .46 & \textbf{1.0} & \textbf{.97} & .33 & \textbf{.72} & .28 & \textbf{1.0} & \textbf{1.0}
\\
\hline
Ising mixture, Component 2&
\textbf{1.0} & .16 & \textbf{1.0} & .18 & .38 & \textbf{1.0} & .22 & \textbf{1.0} & \textbf{1.0} & .16 & {1.0} & \textbf{1.0} & .21 & \textbf{1.0}\\ \bottomrule
\\
\\ \toprule
Model&
$\gamma_{35}$ & 
$\gamma_{36}$ & 
$\gamma_{37}$ &
$\gamma_{38}$ &
$\gamma_{45}$ &
$\gamma_{46}$ & 
$\gamma_{47}$ & 
$\gamma_{48}$ & 
$\gamma_{56}$ &
$\gamma_{57}$ & 
$\gamma_{58}$ & 
$\gamma_{67}$ &
$\gamma_{68}$ & 
$\gamma_{78}$ \\
\hline\\
Ising model &
.46 & \textbf{1.0} & 0.1 & \textbf{1.0} & \textbf{.93} & \textbf{1.0} & \textbf{1.0} & \textbf{1.0} & \textbf{1.0} & \textbf{1.0} & \textbf{1.0} & \textbf{1.0} & \textbf{1.0} & \textbf{1.0}\\
\hline
Ising mixture, Component 1&
\textbf{1.0} & .26 & .18 & \textbf{1.0} & \textbf{1.0} & \textbf{1.0} & \textbf{1.0} & \textbf{1.0} & .27 & \textbf{1.0} & \textbf{1.0} & \textbf{1.0} & \textbf{1.0} & .16
\\
\hline
Ising mixture, Component 2&
.11 & \textbf{1.0} & \textbf{.99} & .15 & .30 & \textbf{1.0} & \textbf{.99} & \textbf{1.0} & \textbf{1.0} & \textbf{1.0} & \textbf{1.0} & \textbf{1.0} & \textbf{1.0} & \textbf{1.0}
\\
\bottomrule 
\end{tabular}}}
\caption{
The posterior means of $\bm\gamma^{(1)}$ and $\bm\gamma^{(2)}$ inferred by two-component Ising mixture models for the NLTCS data.
The number of component in the normal mixture sampling distribution is $J=5$.
The posterior mean of the weight of the first component, i.e. $E(w^{(1)}\mid\bm n)$, is $0.4$.
}
\label{table:posteriorMeanNLTCS}
\end{table}
}

\begin{figure}[ht]
\begin{center}
\includegraphics[width=0.49\textwidth]{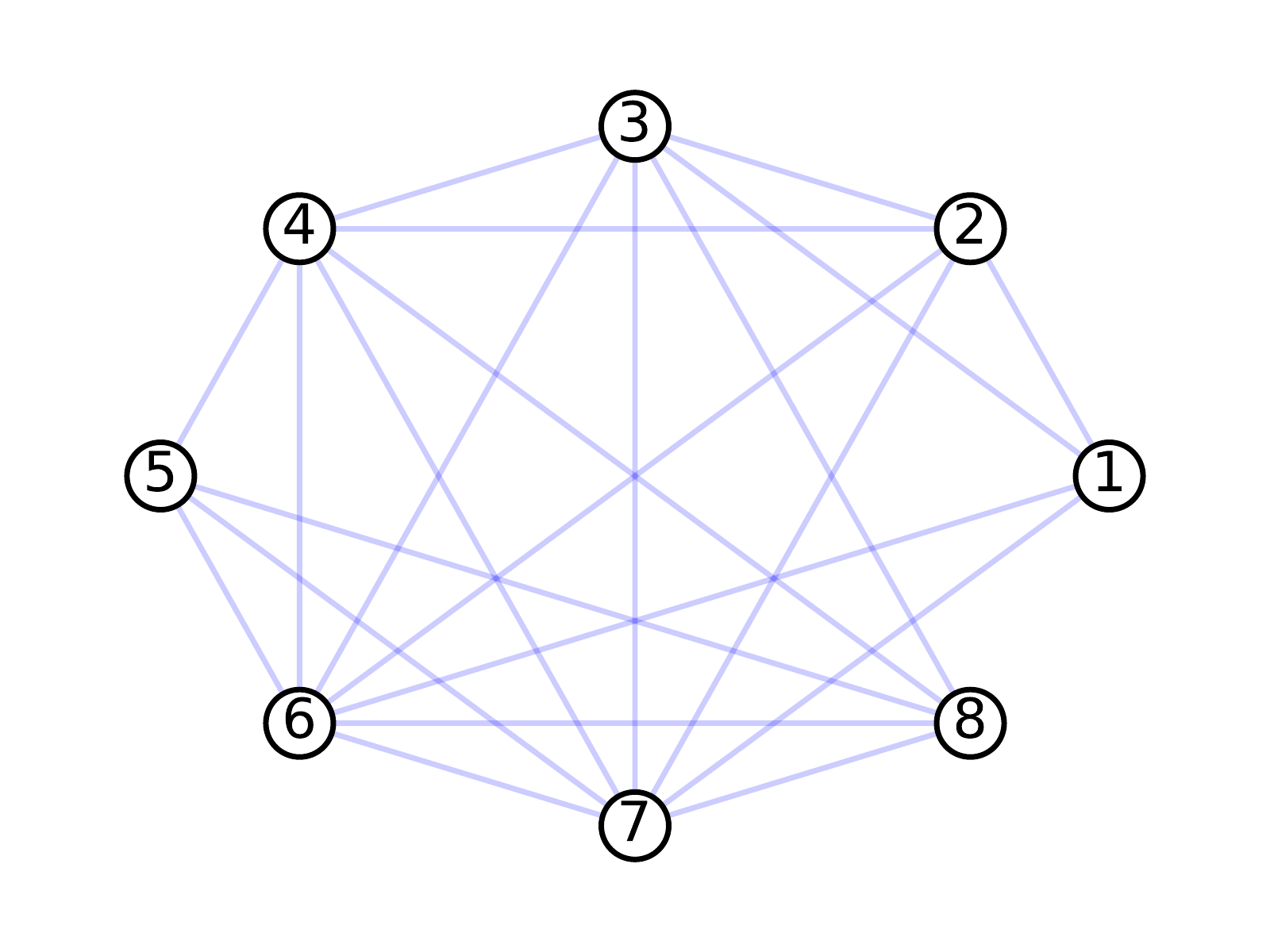}
\includegraphics[width=0.49\textwidth]
{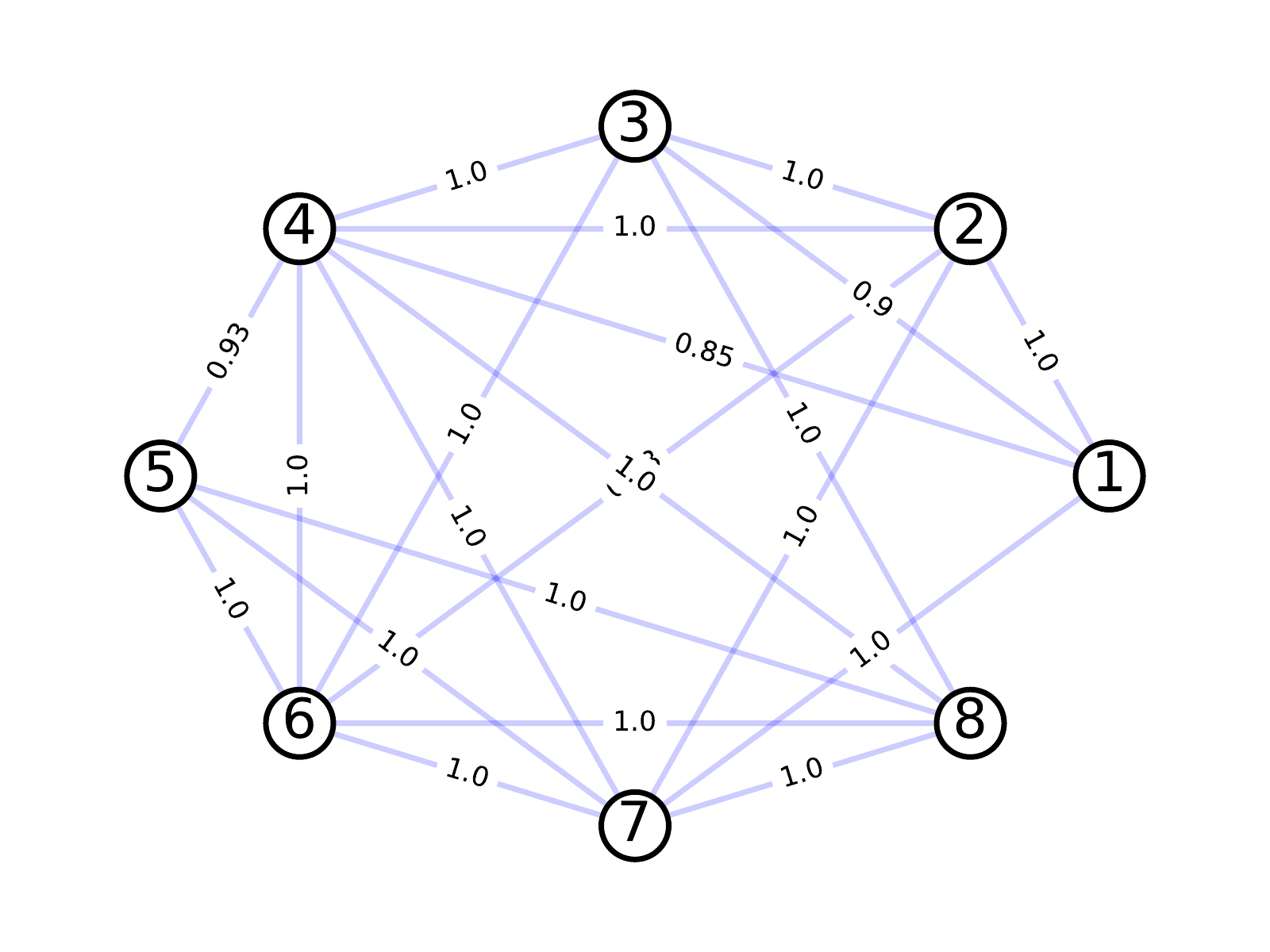}
\end{center}
\caption{
Significant pairwise associations determined by the forward stepwise function based on BIC in the R package gRim \citep{hojsgaard2012graphical} (left panel) and the Bayesian Ising model (right panel) in the NLTCS data. Pairwise associations are shown as edges between vertices associated with variables they involve. The labels of the edges indicate the estimated posterior means of their indicators.
}
\label{fig:NLTCSbeta50IsingEdges}
\includegraphics[width=.49\textwidth]{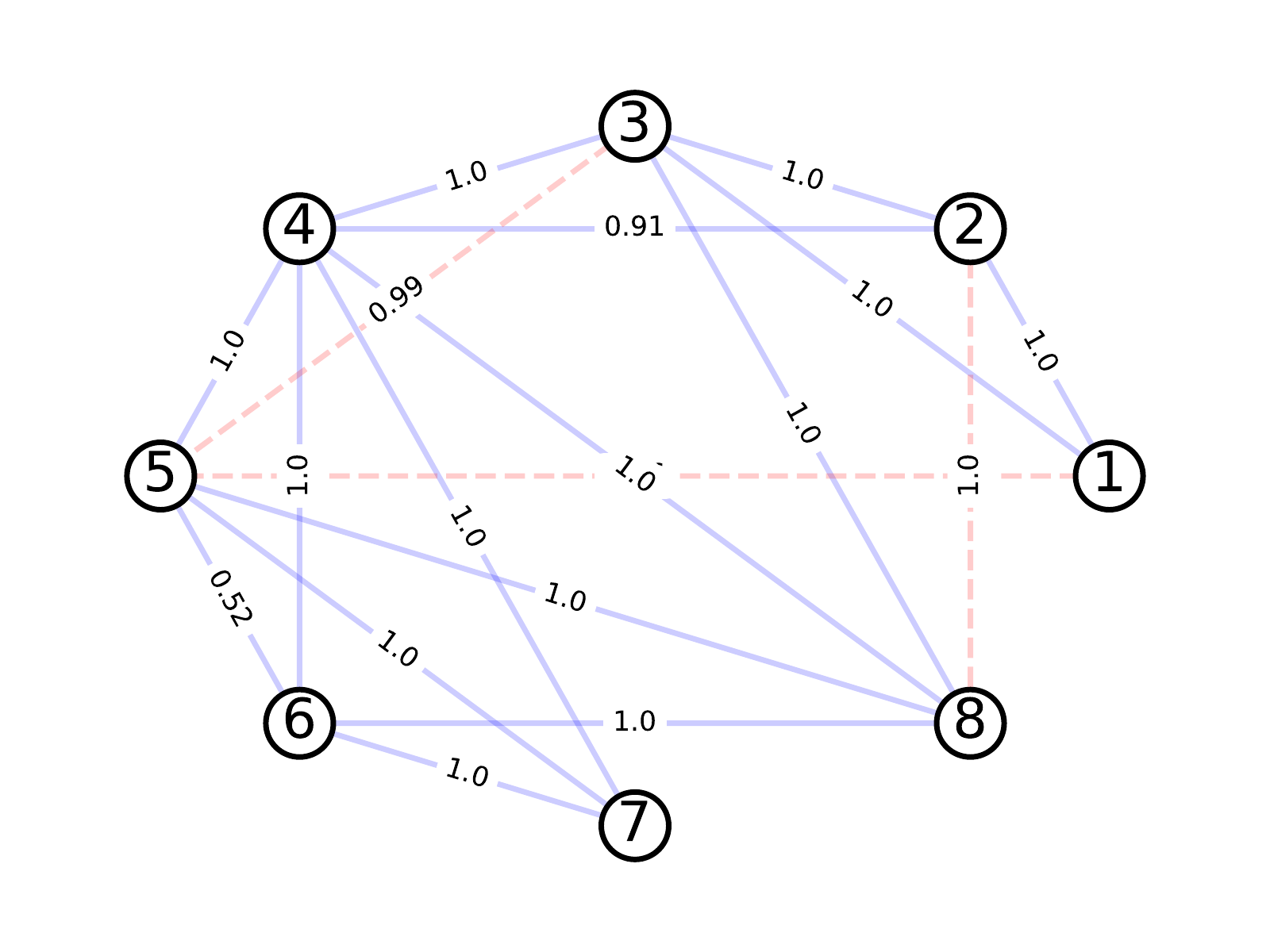}
\includegraphics[width=.49\textwidth]{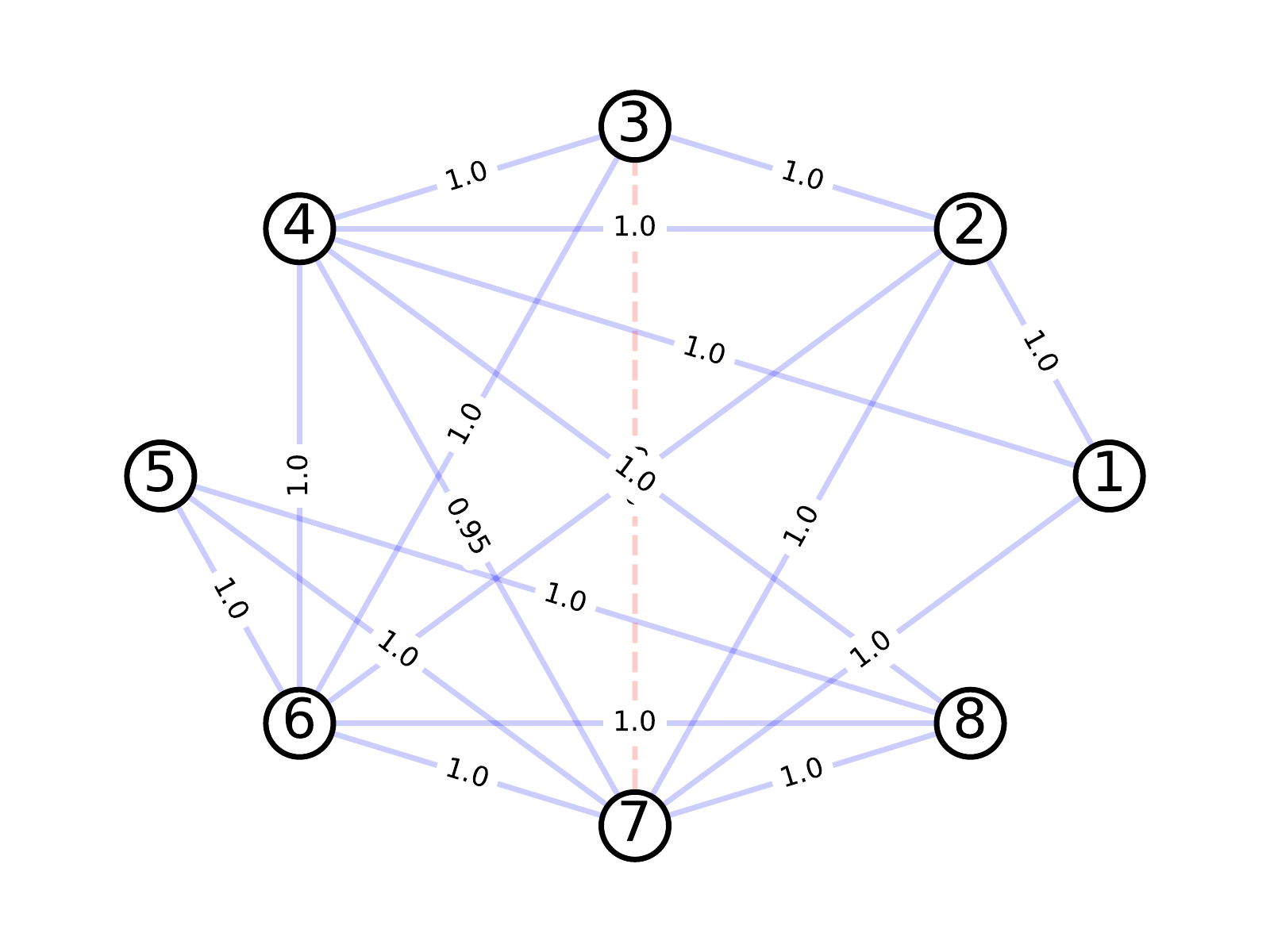}
\caption{Significant pairwise associations determined by the Bayesian Ising mixture model with two components in the NLTCS data. Each association is shown as an edge between vertices associated with variables it involves. There are 18 associations in the first component (shown in the left panel), and 19 associations in the second component (shown in the right panel). The labels of the edges indicate the estimated posterior means of their indicators.
}
\label{fig:NLTCSbeta50IsingMixture}
\end{figure}

\section{Identifiability of Ising mixture models}
\label{subsec:identifiability}

In this section we focus on exploring the identifiability of Ising mixture models from a theoretical perspective. The non-identifiability of the probability mass function given the association indicators, i.e., $\pi(\bm n\mid\bm\Gamma)$, arises from the non-identifiability of the probability mass function of Ising mixture models, i.e., $\pi(\bm n\mid\bm\Theta,\bm w)$. We start with a thorough review of closely related existing results and methods in the literature, and explain why their application to Ising mixture models is challenging. Then we propose some specific sufficient conditions and necessary conditions for the identifiability of Ising mixture models. We also discuss several examples that illustrate specific cases of key interest. Proofs of all the theoretical results are given in the Appendix.

\citet[Corollary 1]{manole2021estimating} provide a sufficient condition for the identifiability of finite mixtures of multinomial distributions. Ising mixture models assume that each cell count follows a multinomial distribution determined by cell probabilities, rather than a mixture of multinomial distributions \--- the situation studied in \citep{manole2021estimating}. Thus their results are not directly applicable in our setting. Other related results in literature are based on the conditional independence assumption, such as \citet{allman2009identifiability} and \citet{xu2017identifiability}. 
In their work, the conditional independence assumption allows for the transfer of the identifiability question to an equivalent one with fewer variables and more levels with the help of the row-wise tensor product. In the case of three categorical variables, the conditional independence assumption further enables the identifiability question to be transformed into an equivalent one by considering the rank of matrices using the triple product. However, these methods are not directly applicable for Ising mixture models since the joint probability conditional on a mixture component cannot be written as a product of marginal probabilities.

\subsection{Examples and main results related to identifiability}
\begin{definition}
An Ising mixture model parameterized by weights $\bm w$ as well as main and interaction terms $\bm\Theta$ is identifiable if and only if different $\bm w$, $\bm\Theta$ imply different cell probabilities $\bm p_{\rm mix}$.
\end{definition}

\begin{definition}[Definition 3 in \cite{rothenberg1971identification}]
A Ising mixture model is locally identifiable at a parameter point $\underline{\bm w},\underline{\bm\Theta}$ if and only if there exists an open neighborhood of $\underline{\bm w},\underline{\bm\Theta}$ containing no other parameter $\bm w,\bm\Theta$ implying the same cell probabilities $\bm p_{\rm mix}$ as $\underline{\bm w},\underline{\bm\Theta}$.
\end{definition}

In the sequel the identifiability is studied based on the following assumption.
\begin{assumption}
The main effects vectors $(\theta^{(k)}_{v}:v\in[d])^T$ are identical for all $k\in[K]$.
\label{assumption:sameMainEffectsForGeneralK}
\end{assumption}
This assumption arises from the similarity among subpopulations. Although it is assumed that each cell probability is a mixture of different components, we don't want to assume that these components are entirely different from each other.
By assuming identical main effects across all components, we can account for the similarity among components and also allow for heterogeneity of interaction effects in different components. Lemma~\ref{fact:withoutlossofgenerality} shows that this assumption simplifies the study of sufficient conditions and necessary conditions, allowing us to bypass main effects and focus on the identifiability of interaction effects.
\begin{lemma}
\label{fact:withoutlossofgenerality}
Under Assumption~\ref{assumption:sameMainEffectsForGeneralK}, 
the identifiability of an Ising mixture model
remains the same when
all main effects are assumed to be $0$.
\end{lemma}

In what follows the parameter vector $\bm\theta$ includes only interaction effects for each component, i.e., $\bm\theta^{(k)}=(\theta_{v^\prime v}^{(k)}:v^\prime <v)^T$ for each $k\in[K]$.

The second assumption arises when interaction effects in only one component are unknown.
\begin{assumption}
All interaction effects in every component except the first component are known. In other words, $\bm\theta^{(k)}$ are fixed and known for all $k\geq2$.
\label{assumption:onlyOneUnkownComponents}
\end{assumption}

However, the following example shows that this assumption is not sufficient for the local identifiability of the Ising mixture model.
\begin{example}
\label{example:loglineard=2}
Suppose $d=2$ and $\theta_{12}^{(2)}=0$.
Then this mixture model is not locally identifiable for any $\theta_{12}^{(1)}\in\mathbb{R}$ and $w^{(1)}\in(0,1)$.\\
The proofs for this example and the following examples are deferred to the appendix.
\end{example}

Therefore, we need another assumption on the weights of components $\bm w$. 
\begin{assumption}
The weights of components $w^{(k)}\in(0,1),k\in[K]$ are fixed and known.
\label{assumption:fixedWeightsForGeneralK}
\end{assumption}

The following proposition states our first sufficient conditions for the local identifiability of Ising mixture models, which is particularly useful when some unknown subpopulation is mixed with other well-known populations.
\begin{proposition}
\label{example:identifiabilityd>=3}
Assumptions ~\ref{assumption:sameMainEffectsForGeneralK},~\ref{assumption:onlyOneUnkownComponents} and ~\ref{assumption:fixedWeightsForGeneralK} are sufficient for local identifiability of Ising mixture models.
\end{proposition}

Next we connect Ising mixture models with mixtures of graphical structures. We represent an Ising model with parameters $\bm\theta$ by an undirected graph $G(\bm\theta):=G(V,\mathbb{E}(\bm\theta))$. In this representation, vertices $V=\{1,2,\ldots,d\}$ are associated with each variable, and edges $\mathbb{E}(\bm\theta):=\{(v^\prime,v):v^\prime<v,\theta_{v^\prime v}\neq0\}$ are associated with each non-zero pairwise interaction. Each missing edge corresponds with a pairwise interaction effect that is zero. The edges in $\mathbb{E}(\bm\theta)$ are called activation edges. We also define $\mathbb{V}(\bm\theta)$ as the set of vertices with degree at least 1. The vertices in $\mathbb{V}(\bm\theta)$ are called activation vertices or activation variables. We define the projection of the graph $G(\bm\theta)$ onto its activation variables $\mathbb{V}(\bm\theta)$ as the subgraph of $G(\bm\theta)$  determined by $\mathbb{V}(\bm\theta)$. This projection is denoted by $G(\bm\theta\mid\mathbb{V}(\bm\theta))$.

The graphical representation of Ising mixture models are constructed at the level of their mixture components. For component $k\in[K]$ with parameters $\bm\theta^{(k)}$, we construct its undirected graph $G(\bm\theta^{(k)})$. The set of activation variables and activation edges for component $k$ are denoted by $\mathbb{V}(\bm\theta^{(k)})$ and $\mathbb{E}(\bm\theta^{(k)})$, respectively.
To illustrate these definitions, consider the following example.

\begin{example}
\label{example:fourVariablesExample1}
Suppose $d=4$, $K=2$, $\theta^{(1)}_{v^\prime v}=0$ for all $(v^\prime,v)\neq(1,2)$ and $\theta^{(2)}_{v^\prime v}=0$ for all $(v^\prime,v)\neq(3,4)$.
Then $G(\bm\theta^{(1)})=(\{1,2,3,4\},\{(1,2)\})$ and $G(\bm\theta^{(2)})=(\{1,2,3,4\},\{(3,4)\})$. The activation variables in component 1 are $\mathbb{V}(\bm\theta^{(1)})=\{1,2\}$. The activation variables in component 2 are $\mathbb{V}(\bm\theta^{(2)})=\{3,4\}$. The activation edges in component 1 are $\mathbb{E}(\bm\theta^{(1)})=\{(1,2)\}$. The activation edges in component 2 are $\mathbb{E}(\bm\theta^{(2)})=\{(3,4)\}$.
$G(\bm\theta^{(1)}\mid\mathbb{V}(\bm\theta^{(1)}))=(\{1,2\},\{(1,2)\})$
and $G(\bm\theta^{(2)}\mid\mathbb{V}(\bm\theta^{(2)}))=(\{3,4\},\{(3,4)\})$. Please see Figure \ref{fig:examle1FourVariables}.

\begin{figure}[ht]
    \centering
\includegraphics[width=0.45\textwidth]{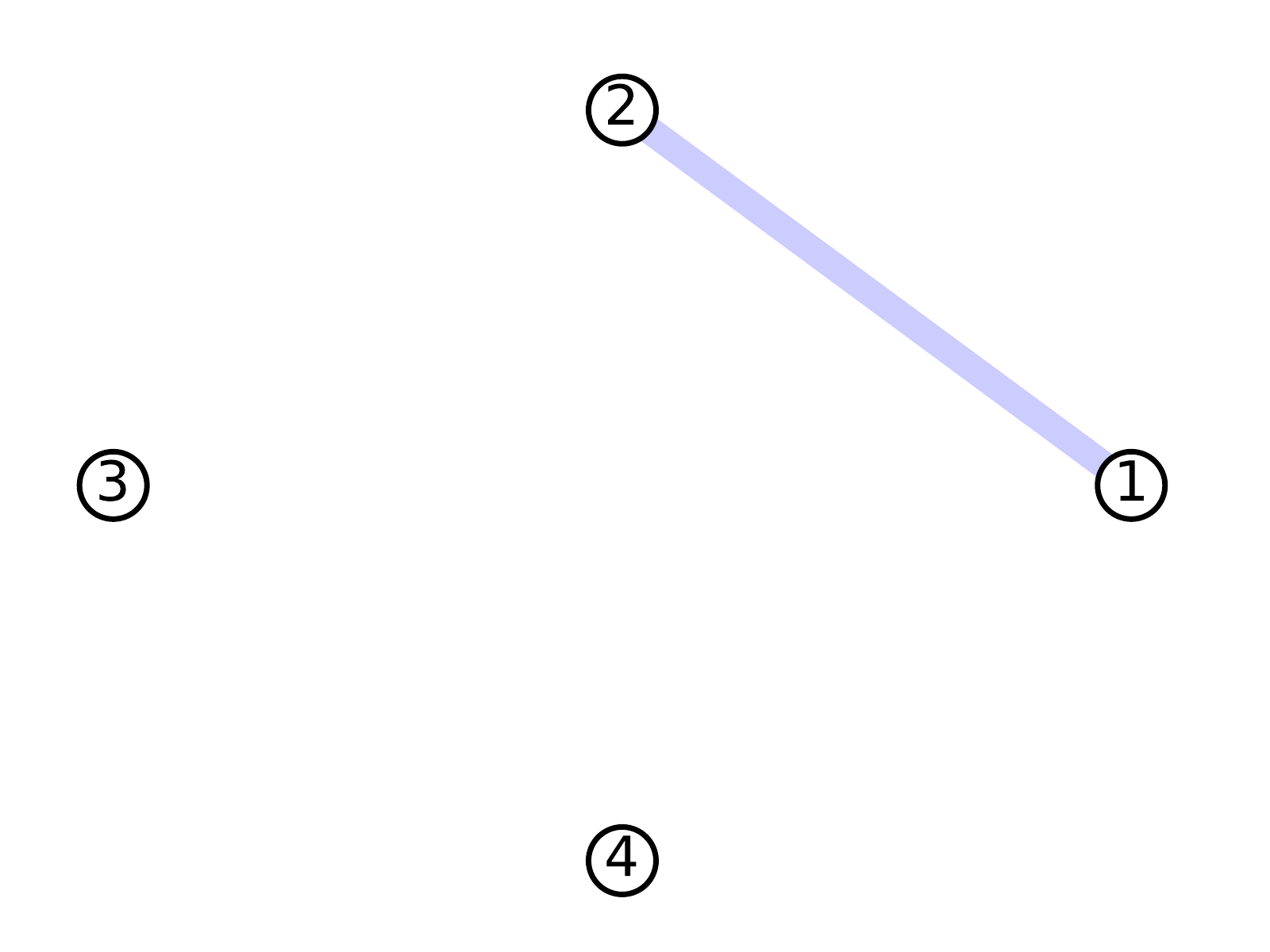}
\includegraphics[width=0.45\textwidth]{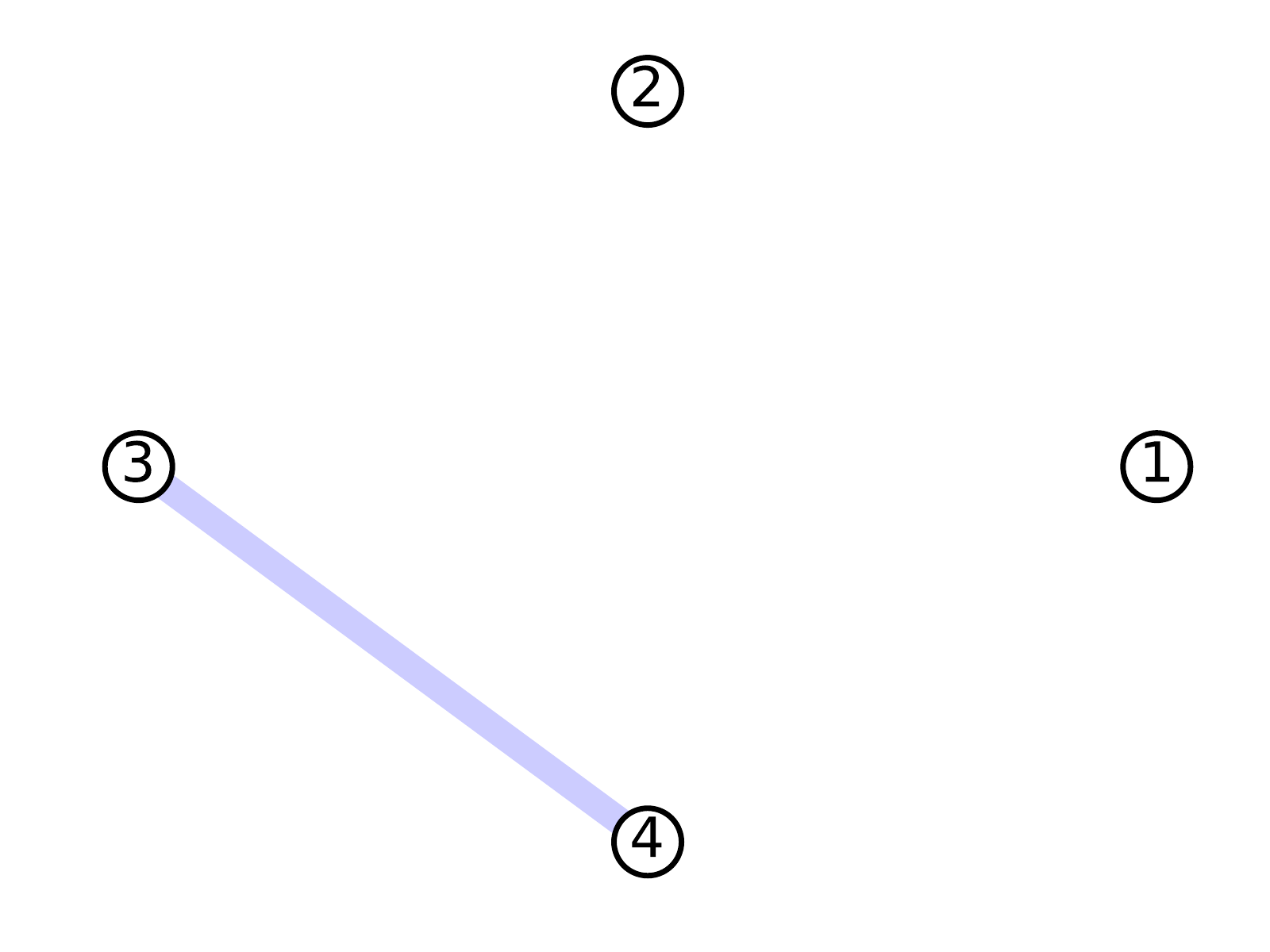}
    \caption{Graph representations for Example~\ref{example:fourVariablesExample1}. Component 1 is shown on the left and component 2 is shown on the right.
     }
\label{fig:examle1FourVariables}
\end{figure}
\end{example}

\noindent
We now are prepared to formulate an essential assumption for identifiability based on graphical representations and activation variables.
\begin{assumption}
The activation variables in different components of an Ising mixture model are mutually exclusive, i.e.,
$\bigcap\limits_{k\in[K]} \mathbb{V}(\bm\theta^{(k)})=\emptyset$.
\label{assumption:disjointVariablesForGeneralK}
\end{assumption}

The next result provides sufficient conditions for identifiability. This result is particularly useful when different components have different activation variables.

\begin{proposition}
\label{prop:disjointIdentifiability}
Assumptions \ref{assumption:sameMainEffectsForGeneralK}, \ref{assumption:fixedWeightsForGeneralK} and \ref{assumption:disjointVariablesForGeneralK} are jointly sufficient for local identifiability of Ising mixture model.
\end{proposition}

The following example illustrates that Assumption~\ref{assumption:disjointVariablesForGeneralK} alone does not guarantee  local identifiability, and therefore Assumption~\ref{assumption:fixedWeightsForGeneralK} is required for the validity of Proposition~\ref{prop:disjointIdentifiability}.

\begin{example}[Assumption~\ref{assumption:fixedWeightsForGeneralK} violated]
\label{example:k=2d=4}
Suppose $K=2$, $d=4$.
Only $\theta_{12}^{(1)}$, $\theta_{34}^{(2)}$ and $w^{(1)}$ are the nonzero unknown parameters. All other interaction effects are fixed at zero. The resulting mixture model is not locally identifiable for any $\theta_{12}^{(1)},\theta_{34}^{(2)}\in\mathbb{R}$ and $w^{(1)}\in(0,1)$.
\end{example}

Based on this example, it can be immediately inferred that any two-component Ising mixture model with at least one non-zero interaction effect in each component is not locally identifiable in general.

The following example shows that Assumption~\ref{assumption:disjointVariablesForGeneralK} is required for the validity of Proposition~\ref{prop:disjointIdentifiability}.
\begin{example}
[Assumption \ref{assumption:disjointVariablesForGeneralK} violated]
\label{example:FixedWeightNotEnough}
Let $d=4$, $w^{(1)}=w^{(2)}=.5$ and $\theta_{13}^{(k)}=\theta_{14}^{(k)}=\theta_{23}^{(k)}=\theta_{24}^{(k)}=0$ for $k=1,2$. The unknown parameters are $\theta_{12}^{(k)},\theta_{34}^{(k)}$ for $k=1,2$ only. Then this mixture model is not locally identifiable.
\end{example}

\section{Discussion}
\label{sec:discussion}

In this paper we developed finite mixtures of Ising within a Bayesian framework as an effective alternative to infer associations between binary variables. By combining Ising models with multivariate Bernoulli mixture models, our contribution addresses the current gap in the literature between loglinear models and various types of mixture models for categorical data. There are several key reasons why addressing this gap was a worthwhile effort. First, Ising mixture models not only effectively fit sparse data, but also offer interpretable results. If data are generated from an Ising mixture model with sparse interaction effects, using Ising models can result in denser and less confident interaction effects. Although Ising mixture models have more parameters than Ising models, they often infer fewer but more significant non-zero interaction effects. This feature of Ising mixture models was illustrated in the simulations experiments. 

Second, Ising mixture models can be viewed as an extension of multivariate Bernoulli mixture models, breaking the conditional independence assumption by introducing interaction effects for each component. Inferring interaction effects from Ising mixture models can lead to the identification of mixtures of graphical loglinear models, providing insight into multivariate patterns of associations of subpopulations. As graphical models become increasingly popular in fields such as social networks, it is likely that Ising mixture models will also gain attention in these areas. Furthermore, the development of Ising mixture models can potentially contribute to the development of more general finite mixtures of graphical loglinear models. 

Ising mixture models are a powerful tool to handle multi-modal spike-and-slab posteriors in data. Research has shown that mixture models can be used to approximate these multi-modal posteriors in linear regressions \citep{rockova2018particle}. Reporting a single model is a misleading reflection of overall model uncertainty. We studied Ising mixture models with spike-and-slab prior distributions to infer associations between binary variables. We have shown that our framework is not only effective in fitting sparse contingency tables, but also leads to interpretable results. We established sufficient and necessary conditions for the identifiability of Ising mixture models without relying on the assumption of conditional independence. More work is certainly needed to propose general conditions for identifiability of Ising mixture models, that will lead to improving the interpretation of the inferred associations and reduce the risk of overfitting.

Our proposed framework can be extended in at least two key directions. It can be generalized to handle categorical random variables with more than two levels starting from the Potts model \citep{RevModPhys.54.235}, although this should be done with care given the increase in the number of parameters. Furthermore, the inclusion of higher-order interaction terms can also be beneficial to allow the study of more complex interaction patterns of associations.

\section{Appendix}
\label{sec:proofs}
\begin{proof}[Proof of Lemma~\ref{fact:withoutlossofgenerality}]
Suppose $\underline{\bm w},\underline{\bm\Theta},\underline{\bm\Gamma}$ are true values of parameters in the Ising mixture model.
It follows from the likelihood equation of $\bm X=\bm 0$ as well as $X_1=1,X_2=\ldots=X_d=0$ that 
\[
\sum_{k\in[K]}\frac{w_k}{Z_k}
=\sum_{k\in[K]}\frac{\underline{w}_k}{\underline{Z}_k}
\text{ as well as }
\sum_{k\in[K]}\frac{w_k\exp(\theta_1)}{Z_k}
=\sum_{k\in[K]}\frac{\underline{w}_k\exp(\underline{\theta}_1)}{\underline{Z}_k},
\]
where $Z_k,\underline{Z}_k,k\in[K]$, are normalization constants.
Dividing the second equation by the first equation, it follows that $\theta_1=\underline{\theta}_1$.
It then follows from analogous arguments that $\theta_k=\underline{\theta}_k$ for all $k\in[K]$.
This lemma then follows immediately from the claim that all likelihood equations with the true value of main effects are equivalent with likelihood equations with main effects $0$.
\end{proof}
\begin{proof}[Proof of Example~\ref{example:loglineard=2}]
Let $\underline{\theta}_{12}^{(1)}$ and $\underline{w}^{(1)}$ be the true value of parameters.
Then we only need to prove that the solutions to the following equations are not unique:
\[
\frac{w^{(1)}}{3+\eta_{12}^{(1)}}+\frac{1-w^{(1)}}{4}=p_{00}=\frac{\underline{w}^{(1)}}{3+\underline{\eta}_{12}^{(1)}}+\frac{1-\underline{w}^{(1)}}{4}
\text{ and }
\frac{w^{(1)}\eta_{12}^{(1)}}{3+\eta_{12}^{(1)}}+\frac{1-w^{(1)}}{4}=p_{11}=\frac{\underline{w}^{(1)}\underline{\eta}_{12}^{(1)}}{3+\underline{\eta}_{12}^{(1)}}+\frac{1-\underline{w}^{(1)}}{4},
\]
where $\eta_{12}^{(1)}=\exp(\theta_{12}^{(1)})$ and $\underline{\eta}_{12}^{(1)}=\exp(\underline{\theta}_{12}^{(1)})$.
It follows from the second equation that 
\[
\eta_{12}^{(1)}=\underline{\eta}_{12}^{(1)}+\frac{(1-\underline{\eta}_{12}^{(1)})(w^{(1)}-\underline{w}^{(1)})}{\frac{\underline{w}^{(1)}}{3+\underline{\eta}_{12}^{(1)}}+\frac{w^{(1)}-\underline{w}^{(1)}}{4}}.
\]
Replace $\eta_{12}^{(1)}$ with this identity in the first equation, it follows that
\begin{align*}
\frac{w^{(1)}}{3+\eta_{12}^{(1)}}+\frac{1-w^{(1)}}{4}
&=\frac{\underline{w}^{(1)}}{3+\underline{\eta}_{12}^{(1)}}+\frac{1-\underline{w}^{(1)}}{4}.
\end{align*}
Therefore, this model is not locally identifiable.
\end{proof}
\begin{proof}[Proof of Proposition~\ref{example:identifiabilityd>=3}]
Let $\underline{\bm w}$ denote the true value of the weights $\bm w$ and it is known by assumptions.
Analogously, let $\underline{\bm\theta}^{(k)}$ denote the true value of interaction effects $\bm\theta^{(k)}$ in the $k$-th component.
The interaction effect in the first component $\bm\theta^{(1)}$ is unknown and its true value is denoted by $\underline{\bm\theta}^{(1)}$.

It then follows from the identity of the cell probabilities, i.e.,
\[
\bm p_{\rm mix}(\bm\theta^{(1)},\underline{\bm w},\underline{\bm\theta}^{(2)},\ldots,\underline{\bm\theta}^{(K)})
=\bm p_{\rm mix}(\underline{\bm\theta}^{(1)},\underline{\bm w},\underline{\bm\theta}^{(2)},\ldots,\underline{\bm\theta}^{(K)})
\]
that
\[
\underline{w}^{(1)}\bm p(\bm\theta^{(1)})+\sum_{2\leq k\leq K}\underline{w}^{(k)}\bm p(\underline{\bm\theta}^{(k)})
=\underline{w}^{(1)}\bm p(\underline{\bm\theta}^{(1)})+\sum_{2\leq k\leq K}\underline{w}^{(k)}\bm p(\underline{\bm\theta}^{(k)}),
\]
and hence
\[
\bm p(\bm\theta^{(1)})
=\bm p(\underline{\bm\theta}^{(1)})
\]
given $\underline{w}^{(1)}>0$.
It then follows from the identifiability of Ising model that $\bm\theta^{(1)}=\underline{\bm\theta}^{(1)}$.
\end{proof}
\begin{proof}[Proof of Proposition~\ref{prop:disjointIdentifiability}]
Without loss of generality we assume all main effects are zero by Lemma~\ref{fact:withoutlossofgenerality}.
Let $\underline{\bm\theta}^{(k)}$ denote the true value of $\bm\theta^{(k)}$ for $k\in[K]$.
The key step is to show that 
\begin{align}
p_{\bm i}(\bm\theta^{(k)})-p_{\bm 0}(\bm\theta^{(k)})=p_{\bm i}(\underline{\bm\theta}^{(k)})-p_{\bm 0}(\underline{\bm\theta}^{(k)})
\label{eq:subtractionEquation1}
\end{align}
for any $\bm i \in I$ and $k\in[K]$.
If these equations hold, we then sum them up over all $\bm i\in I$ and we immediately have $1-2^dp_{\bm 0}(\bm\theta^{(k)})=1-2^dp_{\bm 0}(\underline{\bm\theta}^{(k)})$ or  $p_{\bm 0}(\bm\theta^{(k)})=p_{\bm 0}(\underline{\bm\theta}^{(k)})$.
As a result, we have $p_{\bm i}(\bm\theta^{(k)})=p_{\bm i}(\underline{\bm\theta}^{(k)})$ and hence $\bm\theta^{(k)}=\underline{\bm\theta}^{(k)}$ by the identifiability of Ising models.

It suffices to prove \eqref{eq:subtractionEquation1} for $k=1$. Arguments for remaining cases are analogous and hence omitted. For each $\bm i\in I$, define $\mathbb{V}(\bm i):=\{v:v\in[d],i_v=1\}$ as activated variables for cell $\bm i$.
Then
$G(\bm\theta\mid\mathbb{V}(\bm i))$ is the projection of graph $G(\bm\theta)$ on $\mathbb{V}(\bm i)$ and $G(\bm\theta\mid\mathbb{V}(\bm i)\bigcap\mathbb{V}(\bm \theta))$ is the projection on activated variables $\mathbb{V}(\bm i)$ and activation variables $\mathbb{V}(\bm \theta)$.
The probability of cell $\bm i$ can be written as 
\[
p_{\bm i}(\bm\theta)
=\exp(\sum_{(v^\prime,v)\in \mathcal{E}(G(\bm\theta\mid\mathbb{V}(\bm i)))}\theta_{v^\prime v})/Z
=\exp(\sum_{(v^\prime,v)\in \mathcal{E}(G(\bm\theta\mid\mathbb{V}(\bm i)\bigcap\mathbb{V}(\bm \theta)))}\theta_{v^\prime v})/Z,
\]
where $Z=Z(\bm\theta)$ is a normalization constant depending on $\bm\theta$.

As a consequence, for any fixed $\bm i=(i_1,\ldots,i_d)^T$, 
we have 
\[
p_{\bm i}(\bm\theta^{(1)})=\exp(\sum_{(v^\prime,v)\in \mathcal{E}(G(\bm\theta^{(1)}\mid\mathbb{V}(\bm i)\bigcap \mathbb{V}(\bm\theta^{(1)})))}\theta^{(1)}_{v^\prime v})/Z^{(1)},
\] 
where $G(\bm\theta^{(1)})$ is the undirected graph associated with an Ising model with parameters $\bm\theta^{(1)}$, $\mathbb{V}(\bm i)$ are activated variables for cell $i$, $\mathbb{V}(\bm\theta^{(1)})$ are activation variables for parameters $\bm\theta^{(1)}$, and $Z^{(1)}=Z(\bm\theta^{(1)})$ is the normalization constant.

\noindent
Let $\bm i^{(1)}$ be the cell such that $\mathbb{V}(\bm i^{(1)})=\mathbb{V}(\bm i)\bigcap \mathbb{V}(\bm\theta^{(1)})$ and we immediately have 
\[
G(\bm\theta^{(1)}\mid\mathbb{V}(\bm i^{(1)})\bigcap\mathbb{V}(\bm\theta^{(1)}))
=G(\bm\theta^{(1)}\mid\mathbb{V}(\bm i)\bigcap \mathbb{V}(\bm\theta^{(1)})\bigcap\mathbb{V}(\bm\theta^{(1)}))
=G(\bm\theta^{(1)}\mid\mathbb{V}(\bm i)\bigcap \mathbb{V}(\bm\theta^{(1)})).
\]
Therefore,
we have
\begin{align*}
p_{\bm i^{(1)}}(\bm\theta^{(1)})
&=\exp\Big(\sum_{(v^\prime,v)\in\mathcal{E}(G(\bm\theta^{(1)}\mid\mathbb{V}(\bm i^{(1)})\bigcap\mathbb{V}(\bm\theta^{(1)})))}\theta_{v^\prime v}^{(1)}\Big)/Z^{(1)}\\
&=\exp(\sum_{(v^\prime,v)\in \mathcal{E}(G(\bm\theta^{(1)}\mid\mathbb{V}(\bm i)\bigcap \mathbb{V}(\bm\theta^{(1)})))}\theta^{(1)}_{v^\prime v})/Z^{(1)}\\
&=p_{\bm i}(\bm\theta^{(1)}).
\end{align*}
Now let's consider $p_{\bm i^{(1)}}(\bm\theta^{(k)})$ for any $k\neq 1$. 
Note that $\mathbb{V}(\bm i^{(1)})=\mathbb{V}(\bm i)\bigcap \mathbb{V}(\bm\theta^{(1)})\subset\mathbb{V}(\bm\theta^{(1)})$ and $\mathbb{V}(\bm\theta^{(1)})\bigcap\mathbb{V}(\bm\theta^{(k)})=\emptyset$ by Assumption~\ref{assumption:disjointVariablesForGeneralK}.
We then have $\mathbb{V}(\bm i^{(1)})\bigcap \mathbb{V}(\bm\theta^{(k)})=\emptyset$.
Therefore, for $k\neq 1$ we have
\[
p_{\bm i^{(1)}}(\bm\theta^{(k)})=\exp\Big(\sum_{(v^\prime,v)\in\mathcal{E}(G(\bm\theta^{(k)}\mid\mathbb{V}(\bm i^{(1)})\bigcap \mathbb{V}(\bm\theta^{(k)})))}\theta_{v^\prime v}^{(k)}\Big)/Z^{(k)}=\exp(0)/Z^{(k)}=1/Z^{(k)}.
\]
Then the probability of $\bm X=\bm i^{(1)}$ is 
\[
\underline w^{(1)}p_{\bm i^{(1)}}(\bm\theta^{(1)})+\sum_{k\geq 2}\underline w^{(k)}p_{\bm i^{(1)}}(\bm\theta^{(k)})
=\underline w^{(1)}p_{\bm i}(\bm\theta^{(1)})+\sum_{k\geq 2}\underline w^{(k)}/Z^{(k)}.
\]
This is also true with replacing $\bm\theta^{(k)}$ by its true value $\underline{\bm\theta}^{(k)}$. 
Therefore, the cell probability equation in the case of $\bm X=\bm i^{(1)}$, i.e.,
\[
\underline w^{(1)}p_{\bm i^{(1)}}(\bm\theta^{(1)})+\sum_{k\geq 2}\underline w^{(k)}p_{\bm i^{(1)}}(\bm\theta^{(k)})
=\underline w^{(1)}p_{\bm i^{(1)}}(\underline{\bm\theta}^{(1)})+\sum_{k\geq 2}\underline w^{(k)}p_{\bm i^{(1)}}(\underline{\bm\theta}^{(k)})
\]
 can be written as 
\begin{align}
\underline w^{(1)}p_{\bm i}(\bm\theta^{(1)})+\sum_{k\geq 2}\underline w^{(k)}/Z^{(k)}
=\underline w^{(1)}p_{\bm i}(\underline{\bm\theta}^{(1)})+\sum_{k\geq 2}\underline w^{(k)}/\underline{Z}^{(k)},
\label{eq:cellEquationOfX=i^(1)}
\end{align}
where $\underline{Z}^{(k)}$ is the normalization constant corresponding with $\underline{\bm\theta}^{(k)}$.
Considering the cell probability equation in the case of $\bm X=\bm 0:=(0,\cdots,0)$, it follows from $p_{\bm 0}(\bm\theta^{(k)})=1/Z^{(k)}$ that 
\[
\underline w^{(1)}p_{\bm 0}(\bm\theta^{(1)})+\sum_{k\geq 2}\underline w^{(k)}/Z^{(k)}
=\underline w^{(1)}p_{\bm 0}(\underline{\bm\theta}^{(1)})+\sum_{k\geq 2}\underline w^{(k)}/\underline{Z}^{(k)}.
\]
Then the difference between \eqref{eq:cellEquationOfX=i^(1)} and the identity above implies 
\[
\underline w^{(1)}p_{\bm i}(\bm\theta^{(1)})-\underline w^{(1)}p_{\bm 0}(\bm\theta^{(1)})
=\underline w^{(1)}p_{\bm i}(\underline{\bm\theta}^{(1)})-\underline w^{(1)}p_{\bm 0}(\underline{\bm\theta}^{(1)}),
\]
which gives \eqref{eq:subtractionEquation1}.
\end{proof}
\begin{proof}[Proof of Example~\ref{example:k=2d=4}]
To simplify notations, let $\bm\eta^{(k)}:=\exp(\bm\theta^{(k)})$, $k=1,2$.
Suppose $\underline{\eta}_{12}^{(1)},\underline{\eta}_{34}^{(2)}$ and $\underline{w}^{(1)}$ are true parameters.
It then follows from the likelihood equations that 
\begin{align*}
\frac{w^{(1)}}{3+\eta_{12}^{(1)}}+\frac{1-w^{(1)}}{3+\eta_{34}^{(2)}}
&=\frac{\underline{w}^{(1)}}{3+\underline{\eta}_{12}^{(1)}}+\frac{1-\underline{w}^{(1)}}{3+\underline{\eta}_{34}^{(2)}},\\
\frac{w^{(1)}\eta_{12}^{(1)}}{3+\eta_{12}^{(1)}}+\frac{1-w^{(1)}}{3+\eta_{34}^{(2)}}
&=\frac{\underline{w}^{(1)}\underline{\eta}_{12}^{(1)}}{3+\underline{\eta}_{12}^{(1)}}+\frac{1-\underline{w}^{(1)}}{3+\underline{\eta}_{34}^{(2)}},\\
\frac{w^{(1)}}{3+\eta_{12}^{(1)}}+\frac{(1-w^{(1)})\eta_{34}^{(2)}}{3+\eta_{34}^{(2)}}
&=\frac{\underline{w}^{(1)}}{3+\underline{\eta}_{12}^{(1)}}+\frac{(1-\underline{w}^{(1)})\underline{\eta}_{34}^{(2)}}{3+\underline{\eta}_{34}^{(2)}}.
\end{align*}
From the first two equations it follows that 
\[
\eta_{12}^{(1)}
=\underline{\eta}_{12}^{(1)}
+\frac{(w^{(1)}-\underline{w}^{(1)})(1-\underline{\eta}_{12}^{(1)})(\underline{\eta}_{12}^{(1)}+3)}{w^{(1)}(\underline{\eta}_{12}^{(1)}+3)+\underline{w}^{(1)}(1-\underline{\eta}_{12}^{(1)})},
\]
where $\eta_{12}^{(1)}$ is represented as a function of $w^{(1)}$.
Analogously, it follows from the first and the third equations that
\[
\eta_{34}^{(2)}
=\underline{\eta}_{34}^{(2)}
-\frac{(w^{(1)}-\underline{w}^{(1)})(1-\underline{\eta}_{34}^{(2)})(\underline{\eta}_{34}^{(2)}+3)}{(1-w^{(1)})(\underline{\eta}_{34}^{(2)}+3)+(1-\underline{w}^{(1)})(1-\underline{\eta}_{34}^{(2)})},
\]
where $\eta_{34}^{(2)}$ is represented as a function of $w^{(1)}$.
After some algebra we can see that $\eta_{12}^{(1)}$ and $\eta_{34}^{(2)}$ satisfy the first equation without other constraints. In other words, we can construct different values of parameters based on $w^{(1)}\in(0,1)$ with the same cell probability.
\end{proof}
\begin{proof}[Proof of Example~\ref{example:FixedWeightNotEnough}]
In this example, the nonidentifiability is justified by the singularity of the information matrix based on the equivalence between the local identifiability and the rank of the information matrix, as established in \citet{rothenberg1971identification,catchpole1997detecting}.
Some examples of Ising mixture models with singular information matrices are in Table~\ref{tab:exampleOfnecessaryCondition1}.
\begin{table}[ht]
\centering
\begin{tabular}{|c|c|c|c|c|c|}
\hline
Index 	& $\theta^{(1)}_{12}$ 	& $\theta^{(1)}_{34}$  	& $\theta^{(2)}_{12}$ 	& $\theta^{(2)}_{34}$ &eigenvalues of its Fisher information matrix\\
1 		&0.5					&0.5					&0.1					&0.1	&0.118389, 0.116171, 0.00220142, 3.04905e-17\\
2 		&0.5					&0.5					&0.1					&1	&0.117381, 0.102194, 0.00242613, -9.52222e-17\\
3 		&3					&0.5					&0.1					&1	&0.104261, 0.084175, 0.0228727, 1.59007e-16\\ \hline
\end{tabular}
\caption{Examples of Ising mixture models with singular information matrices.}
\label{tab:exampleOfnecessaryCondition1}
\end{table}
\end{proof}

\bibliographystyle{apalike}
\bibliography{paper_v1}
\end{document}